\documentclass[11pt]{article}
\usepackage{zzt}

\title{Spectral Methods for the Complexity of Planar Graph Homomorphisms}
\author{
Jin-Yi Cai\thanks{Department of Computer Sciences,
University of Wisconsin-Madison}\\
\texttt{jyc@cs.wics.edu}
\and Ashwin Maran\footnotemark[1]\\
\texttt{amaran@wisc.edu}
\and Zhuxiao Tang\footnotemark[1]\\
\texttt{zztang@wisc.edu}
}
\date{}

\begin{document}

\maketitle

\begin{abstract}
We explore the frontier beyond the recently discovered barrier represented by    
the \emph{quantum automorphism group} $\qut(M)$~\cite{ 
wang1998quantum, banica2005quantum,
lupini2020nonlocal, manvcinska2020quantum,cai2026dichotomy}
%%% JYC: I made it chronological
in the classification theory of planar
graph homomorphisms $\PlGH(M)$.
We show that analyzing the spectral relations of $M$ can 
%{\it analytic} methods by analyzing the spectral relations of $M$ can be used to
prove \#P-hardness when traditional vertex separation and domain-reduction methods with planar edge gadgets 
provably fail due to the $\qut(M)$ barrier.
We prove two criteria of \#P-hardness for $\PlGH(M)$: a spectral criterion and a determinant criterion.
%The proof makes use of matrix perturbation theory.
It is known that the core problem for the classification of $\PlGH(M)$ for nonnegative matrices $M$ is for 
positive definite entry-wise positive matrices~\cite{cai2023complexity,cai2024polynomial,cai2026dichotomy}.
We  use the spectral criterion to show that
 $\PlGH(M)$  is  \#P-hard for all  circulant matrices of prime order 
$q \ge 3$, while for $q=2$ it is precisely the matchgate case and
is P-time computable by the FKT algorithm (for planar perfect matching).
We also prove a complexity dichotomy for $\PlGH$ problems defined by tensor products of 2 by 2 matrices. 
This gives a complete complexity classification for this class of
matrices, and the FKT algorithm together with a holographic transformation is \emph{universal}---every $\PlGH(M)$ is either (1) P-time computable over all graphs, or (2) \#P-hard in general but  P-time computable over planar graphs, or (3) \#P-hard over planar graphs; furthermore, $\PlGH(M)$ in (2) consists of precisely
those computable by FKT with a holographic transformation.
\end{abstract}
\thispagestyle{empty}

\newpage

\tableofcontents

\thispagestyle{empty}
\clearpage
\pagenumbering{arabic}

\section{Introduction}\label{sec: introduction}

Given graphs $G$ and $H$,
a mapping from $V(G)$ to $V(H)$ is called
a \emph{homomorphism}
if the edges of $G$ are mapped to the edges of $H$.
More generally,
let $M = (M_{ij})$ be a $q \times q$ symmetric matrix
with entries $M_{ij} \in \mathbb{R}_{\geq 0}$.
We interpret $M$ as defining a weighted graph $H = H_M$ on $[q]$,
where $M_{ij}$ is the weight of edge $(i,j)$.
Given $M$, the associated \emph{partition function}
$Z_{M}(G)$ for an input undirected multigraph $G = (V,E)$ is
$Z_{M}(G) = \sum_{\sigma: V \rightarrow [q]} \prod_{(u, v) \in E} M_{\sigma(u) \sigma(v)}.$

Isomorphic graphs $G \cong G'$ have the same value
$Z_{M}(G) = Z_{M}(G')$, thus every $M$
defines a graph property $Z_{M}(\cdot)$.
For a $0$--$1$ matrix $M$,
$Z_{M}(G)$ counts the number of homomorphisms from $G$ to $H$.
Graph homomorphism ($\GH$) encompasses a great deal of
graph properties and counting problems
arising in combinatorics and statistical physics
~\cite{lovasz1967operations, lovasz2012large, hell2004graphs}.

Each $M$ defines a
computational problem denoted by $\GH(M)$: given an input graph $G$,
output $Z_{M}(G)$.
The  complexity of 
$\GH(M)$ has been a major focus of research.
A number of increasingly general complexity dichotomy theorems have been achieved~\cite{dyer2000complexity,bulatov2005complexity,goldberg2010complexity,cai2013graph}.
A central feature of these results is the use of
\emph{edge gadgets}:
%small graphs
%%% JYC: doesn't have to be small
graph fragments that replace edges of the input instance
and enable polynomial-time reductions.
%Graph homomorphism 
$\GH$ can also be viewed as
a special case of counting $\CSP$, where
%a series of results established 
a complexity dichotomy is proved for
any set of constraint functions ${\cal F}$
~\cite{bulatov2013complexity,dyer2010complexity,dyer2011decidability,dyer2013complexity,bulatov2012complexity,cai2016nonnegative,cai2017complexity}.

Parallel to this development,
Valiant~\cite{valiant2008holographic}
introduced \emph{holographic algorithms}.
It is well known that counting the number of perfect
matchings (\#PM) is \#P-complete~\cite{valiant1979complexity}. 
On the other hand,
%it is known 
%since the 60's 
%that
the
famous FKT algorithm~\cite{kasteleyn1961statistics,temperley1961dimer,kasteleyn1963dimer,kasteleyn1967graph} from the 60's
 can
compute \#PM on {\it planar graphs} in P-time.
%Valiant's h
Holographic algorithms
greatly extended its reach.
%in fact so much so that a 
A most intriguing question arises:
Is this a \emph{universal} algorithm that 
\emph{every} counting problem expressible as a sum-of-products
that \emph{can be solved} in P-time on planar graphs
(but \#P-hard in general)
\emph{is solved} by this method alone?

After a series of work~\cite{cai2009holant,cai2016complete, backens2017new, backens2018complete, yang2022local, fu2019blockwise, fu2014holographic, cai2019holographic}
it was established that for every set of complex valued constraint functions ${\cal F}$
on the Boolean domain (i.e., domain size $q = 2$) there is a 3-way exact
classification for  
\#$\CSP$(${\cal F}$): 
 (1) P-time solvable, 
 (2) P-time solvable over planar graphs but \#P-hard over general graphs, 
 (3) \#P-hard over planar graphs.
Moreover,
category (2) consists of precisely  those problems that can be solved by
Valiant's holographic algorithm using FKT.

Extending this understanding to larger domains has proven remarkably difficult
\cite{cai2023complexity,cai2024polynomial},
even in the restricted setting of graph homomorphisms.
Hardness proofs in the planar setting typically rely on
\emph{planar edge gadgets},
which must preserve planarity while enabling reductions.
Even for $q = 3$ and $q = 4$,
hardness results were only possible 
through the construction of
individually tailored gadgets and 
highly specific analyses.
%on a case-by-case basis.

To address this, the paper~\cite{cai2026dichotomy}
initiated a systematic study of these planar edge gadgets
to understand the fundamental challenges encountered
when moving beyond the $q = 2$ case.
Let $\PlGH(M)$ denote the problem $\GH(M)$ when the
input graphs $G$ are restricted to planar graphs.
\cite{cai2026dichotomy} established that
$\PlGH(M)$ is \#P-hard whenever
planar gadgets can be used to \emph{separate} the
diagonal entries of $M$.
%effectively distinguishing between different labels.
We call this \emph{vertex separation} method.
However, they also discovered a formidable barrier:
for many matrices, such separation is impossible.
This occurs when the matrix possesses symmetries described by
an abstract construct called the \emph{quantum automorphism group} $\qut(M)$~\cite{wang1995free, wang1998quantum,
bichon2003quantum, banica2005quantum,
atserias2019quantum, chassaniol2019study,
lupini2020nonlocal, manvcinska2020quantum, kar2026npa}.
%%% JYC i edited a bit here. i don't want people to think this has to do literally with quantum computing
% or any change of model of computation, thus P and #P-hardness is not the classical established
% notion. this is especially the case since we had a section called model of computation.
When this  $\qut(M)$ is nontrivial, the diagonal
entries remain inseparable by
\emph{all} planar gadgets, creating a ``quantum gap''
%%% JYC I added a word here "all"
where the standard strategy of finding reductions
from smaller domain \#P-hard problems
becomes provably impossible
(see \autoref{Appendix: Quantum}).
%for more details
%%% ZZT: this makes 11 pages to 10 pages

This paper explores the frontier
beyond this barrier.
While this quantum group barrier exposes fundamental limitations of 
combinatorial vertex separation,
it does not render planar edge gadgets useless.
Rather, it suggests that progress toward a final dichotomy
requires more sophisticated
proof techniques
that operate \emph{within}
the constraints of these quantum symmetries.
In this paper, we show that planar edge gadgets remain
a potent tool, coupled with \emph{analytic} methods, 
for proving hardness, even in cases where 
traditional separation and domain-reduction methods fail.

We briefly outline our analytic approach.%
%%%ZZT: I add a paragraph here
\footnote{
%Introducing analysis logically involves
This introduces 
%Although the asymptotic analysis is reasoning about 
a continuum of matrices which are
%serve as 
\emph{names} of (logically not all distinct) computational problems, while in  strict TM computability  there are  at most countably many distinct problems. This  is not a bug but  a feature; see \autoref{appendix: modelComputation}. 
%for more details about the computational model.
An innovative aspect of working in this setting is
%with a continuum  of  problem (names)
to allow our analytic proof to be carried out.}
%We do pay a price: for \#P-hardness reductions are proved to exist, they are proved non-constructively.}
%See \autoref{appendix: modelComputation} for more details about the computational model.}
Various planar edge gadgets represent
specific algebraic operations on $M$.
Using polynomial interpolation,
we identify matrix-valued entrywise-analytic functions
$\mathcal{G}$, such that
$\PlGH(\mathcal{G}(\cdot)) \leq_{p}^{T} \PlGH(M)$.
%%% JYC this dot doesn't seem right. as is, on the left technically there is no
%%% computational problem specified. perhaps you meant 
%$\PlGH(\mathcal{G}(M)) \leq_{p}^{T} \PlGH(M)$, where \mathcal{G} can be a class of
% different functions.
We then show that the lattice of multiplicative relations
among eigenvalues provides a sufficient hardness certificate of $\PlGH(\cdot)$.
We represent any such relation of 
%$M$ 
$\mathcal{G}(\cdot)$
by defining
a scalar-valued analytic function $f(\mathcal{G}(\cdot))$.
Using asymptotics, we show that $f(\mathcal{G}(\cdot))$
is not identically zero.
Since a nonzero analytic function has only countably many zeros,
we can avoid countably many ``problematic'' spectral relations of $M$ to prove \#P-hardness.
Our proof also utilizes several algebraic techniques
%, 
%including properties of the cyclotomic extension $\mathbb{Q}(e^{2\pi\mathfrak{i}/q})$ over $\mathbb{Q}$ when $q$ is prime.
including 
%generating sets and a property of the 
cyclotomic fields.

Our first result identifies two
useful criteria for proving
\#P-hardness for matrices with
non-trivial quantum symmetries.
A key insight from~\cite{cai2026dichotomy} is that
any universal proof of hardness for all matrices
that uses planar edge gadgets to find reductions within $\PlGH$ problems,
must necessarily involve reductions from the
\emph{Potts model} --
a highly symmetric matrix where all 
non-diagonal entries are equal.
They are the only \#P-hard matrices that
admit the maximal quantum automorphism group,
making them the only viable candidate source
for planar edge gadget reductions that
are universally applicable.
We broaden this framework by showing that 
hardness is determined by a spectral criterion:
as long as the \emph{lattice} of homogeneous multiplicative relations among eigenvalues of $M$ (we call this \emph{spectral lattice}) is comparable to that of the Potts model,
the problem is \#P-hard (\cref{thm: only look at lattice of eigenvalues}).
As an application, we prove an easy-to-test determinant criterion (\cref{thm: determinant criterion}) for \#P-hardness, utilizing \cref{thm: only look at lattice of eigenvalues} and matrix perturbation theory. 
%Although not a necessary condition, this easy-to-test condition is compatible with quantum symmetries, and applies in cases where vertex separation is impossible.

Our second result concerns 
\emph{circulant matrices},
which are defined by their cyclic symmetry.
Circulant matrices are diagonal
identical, and firmly situated within the ``quantum gap,'' ruling out vertex separation techniques.
%to prove $\#$P-hardness. 
%%%ZZT: I think this is not accurate.
Nonetheless, we show that
$\PlGH(M)$ is \#P-hard for positive definite circulant matrices of prime order $q \ge 3$  (\cref{thm: circulant prime order}).
%where $q \geq 3$ is prime (\cref{thm: circulant prime order}).
We analyze the spectral relations of the matrix under 
%a sufficiently 
a large thickening, then 
%make use of 
use \cref{thm: only look at lattice of eigenvalues} to reduce from the Potts model.
%By reducing from a hard matrix of the same size, we demonstrate that planar edge gadgets \emph{can} prove hardness even for matrices possessing non-trivial quantum symmetries, when traditional vertex separation reductions are completely blocked.

% Our second result investigates
% rank-1 perturbations of the identity matrix.
% A key insight from~\cite{cai2026dichotomy} is that
% any universal proof of hardness for all matrices
% must necessarily involve reductions from the
% \emph{Potts model} --
% a specific, highly symmetric matrix where all 
% non-diagonal entries are equal.
% They are the only \#P-hard matrices that
% admit the maximal quantum automorphism group,
% making them the only viable candidate source
% for a reduction that is universally applicable.
% We broaden this framework by showing that 
% hardness is determined by a spectral criterion:
% as long as the eigenvalues of the matrix
% generate a specific algebraic structure (a lattice) 
% comparable to that of the Potts model,
% the problem is \#P-hard.
% Although not a necessary condition,
% this easy-to-test condition applies more broadly
% than other previously popular
% conditions such as vertex separation,
% and is compatible with quantum symmetries.

While the spectral criterion we prove is \emph{sufficient},
it is not \emph{necessary}.
When
$M = A_1 \otimes \cdots \otimes A_s$,
where each 
tensor factor is a $2 \times 2$ matrix,
it is seen that the spectral lattice of $M$ is highly non-trivial %and highly structured, 
and impervious to planar gadgets.
So, our spectral criterion is insufficient
for proving \#P-hardness, even when
$\PlGH(M)$ is already known to be \#P-hard (e.g., for $4 \times 4$ matrices~\cite{cai2024polynomial}).
Nevertheless, we use the analytic approach to separate the normalized eigenvalues of factors $A_i$, thus realizing a trivial ``reduced lattice" (\cref{lm: rho has no relation after thickening}).
This allows us to reduce $\PlGH(M)$ from $\PlGH(A_i^{\otimes t_i})$ for any factor 
$A_i$ with multiplicity $t_i$ (\cref{lm: reduced lattice lemma}, proof in \autoref{Appendix: Tensor of 2by2}).
We prove that $\PlGH(M)$ is \#P-hard unless each factor $A_i$ has equal diagonal entries, in which case the problem is solvable via Valiant's holographic algorithm plus the FKT algorithm (\cref{thm: hardness for 2 by 2 tensor positive definite}). 
Our result therefore isolates a sharp boundary case, namely when $M$ is a tensor product of near matchgates.
In this instance, both the vertex separation and spectral criterion fail, and 
\#P-hardness occurs on one side of this precise boundary but P-time tractability occurs
on the other side.

Collectively, these results demonstrate that planar edge gadgets, 
when combined with analytic and spectral methods,
can push forward the frontier 
toward a universal dichotomy for planar graph homomorphisms. 
Full versions of all omitted or abbreviated proofs can be found in the appendices.

\section{Preliminaries}\label{sec: preliminaries}

Let $\Sym{q}{}{X}$ denote the set of
$q \times q$ symmetric matrices with entries from
$X \subseteq \mathbb{R}$.
For example, we can have $X = \mathbb{R}$, $\mathbb{R}_{\geq 0}$ or
$\mathbb{R}_{\neq 0}$.
We then let $\Sym{q}{F}{X} \subset \Sym{q}{}{X}$ 
and $\Sym{q}{pd}{X} \subset \Sym{q}{F}{X}$ denote, respectively, the
subsets of full rank and positive definite symmetric matrices.
We justify this decision to consider arbitrary real valued
matrices (rather than just algebraic real valued matrices)\footnote{However, our results stay strictly
 within the classical Turing machine model in terms of  bit-complexity.}
in \autoref{appendix: modelComputation}.
We use $\leq^{T}_{p}$ and $\equiv_p^T$ to denote polynomial-time Turing reducibility and P-time Turing equivalence, respectively.
We start with a trivial
lemma which shows that scaling $M$ has no effect on the complexity of $\PlGH(M)$:

\begin{restatable}{lemma}{MequivalentCM}\label{lemma: MequivalentCM}
    Let $M \in \Sym{q}{}{\mathbb{R}}$, and
    let $c \in \mathbb{R}_{\neq 0}$.
    Then $\PlGH(M) \equiv_p^T \PlGH(cM)$.
\end{restatable}

A \emph{planar edge gadget} is a 2-labeled
planar graph $\mathbf{K} = 
(V(\mathbf{K}), E(\mathbf{K}))$
with distinguished vertices $\ell_{1} \neq \ell_{2}
\in V(\mathbf{K})$ on the outer face of $\mathbf{K}$.
Given any $M \in \Sym{q}{}{\mathbb{R}}$,
we define the \emph{signature} $\mathbf{K}(M) 
\in \mathbb{R}^{q \times q}$ as:
\begin{equation}\label{equation: signature}
    \mathbf{K}(M)_{ij} = \sum_{\substack{\tau: 
    V(\mathbf{K}) \rightarrow [q]\\
    \tau(\ell_{1}) = i, \tau(\ell_{2}) = j}}
    \prod_{(u, v) \in E(\mathbf{K})}M_{\tau(u)\tau(v)}.
\end{equation}

In this work, we consider planar edge gadgets whose signatures are symmetric matrices.
For $M \in \Sym{q}{}{\mathbb{R}}$,
let $\PlEdge(M) := \{\mathbf{K}(M)\mid 
\text{planar edge gadget } \mathbf{K}\} \cap \Sym{q}{}{\mathbb{R}}$.

Given any planar graph $G = (V, E)$ and $\mathbf{K}$ with $\vk(M) \in \PlEdge(M)$, we construct the planar graph $\mathbf{K}G = (V(\mathbf{K}G), E(\mathbf{K}G))$
by replacing every
edge $(u, v) \in E$ with a copy of
$\mathbf{K}$, and identifying $u, v$ with
$\ell_{1}, \ell_{2}$ respectively.
Since $\vk(M)$ is symmetric, it doesn't matter how we name each edge $\{u, v\}$ as $(u,v)$ or $(v,u)$
when we replace it by a copy of
$\mathbf{K}$,
the following $Z_M(\vk G)$ is well-defined:
%Technically, $\vk G$ is
%not yet well-defined as a graph, as its definition depends on
%the orientation
%how we name each edge $\{u, v\}$ as $(u,v)$ or $(v,u)$;
%of each edge, 
%however, by the symmetry of $\vk(M)$, the value $Z_M(\vk G)$ will not depend on 
%these edge orientations. We have

\begin{equation} \label{eq:signature_matrix}
    Z_M(\vk G)
    = \sum_{\sigma: V \to [q]} \prod_{(u,v) \in E} 
    ~\sum_{\substack{\tau: V(\vk) \to [q] \\ \tau(\ell_1) = \sigma(u),  \tau(\ell_2) = \sigma(v)}}~
    \prod_{(u',v') \in E(\vk)} M_{\tau(u')\tau(v')}
    = Z_{\vk(M)}(G).
\end{equation}

%%% JYC: maybe should use $\leq_T$ or $\leq^p_T$.
Then $\PlGH(\vk(M)) \leq^{T}_{p} \PlGH(M)$ 
%(Here $\leq$ denotes P-time reduction) 
for all
$\vk(M) \in \PlEdge(M)$.
Some examples of planar edge gadgets follow:
\subsection{Thickening Gadgets}\label{sec: thickeningGadgets}

For an integer $n \geq 1$,
the \emph{thickening gadget} $\mathbf{T_{n}}$ consists
of two vertices $\ell_{1}, \ell_{2}$,
that are connected by $n$ parallel edges
(see \cref{fig:thickening}).

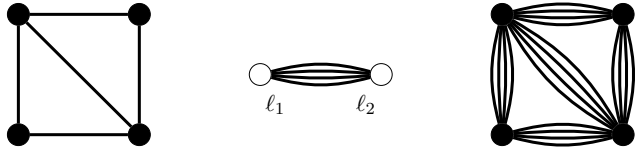
\begin{figure}[htbp]
	\centering
	\scalebox{0.8}{\begin{tikzpicture}[line join=miter, draw opacity=1]

    % Define nodes
    \node[circle, draw=black, fill=black] (A) at (-1, -1){};
    \node[circle, draw=black, fill=black] (B) at (-1, 1){};
    \node[circle, draw=black, fill=black] (C) at (1, 1){};
    \node[circle, draw=black, fill=black] (D) at (1, -1){};
    
    % --- (A--B): one edge ---
    \draw[line width=0.5mm, black] (A) -- (B);
    
    % --- (B--C): one edge ---
    \draw[line width=0.5mm, black] (B) -- (C);
    
    % --- (C--D): one edge ---
    \draw[line width=0.5mm, black] (C) -- (D);
    
    % --- (D--A): one edge ---
    \draw[line width=0.5mm, black] (D) -- (A);
    
    % --- (B--D): one edge ---
    \draw[line width=0.5mm, black] (B) -- (D);

\begin{scope}[xshift = 4cm]

    % Define nodes
    \node[circle, draw=black, fill=white] (A) at (-1, 0){};
    \node[circle, draw=black, fill=white] (B) at (1, 0){};

    % Label nodes
    \node[draw=none] at (-0.75, -0.5) {$\ell_{1}$};
    \node[draw=none] at (0.75, -0.5) {$\ell_{2}$};
    
    % --- (A--B): two evenly spaced arcs ---
    \foreach \angle in {5, 15} {
      \draw[line width=0.5mm] (A) to[bend left=\angle] (B);
      \draw[line width=0.5mm] (A) to[bend right=\angle] (B);
    }
    
\end{scope}

% \begin{scope}[xshift = 4cm]

    % % Define nodes
    % \node[circle, draw=black, fill=black] (A) at (-1, -1){};
    % \node[circle, draw=black, fill=black] (B) at (-1, 1){};
    % \node[circle, draw=black, fill=black] (C) at (1, 1){};
    % \node[circle, draw=black, fill=black] (D) at (1, -1){};
    
    % % --- (A--B): two evenly spaced arcs ---
    % \foreach \angle in {15} {
    %   \draw[line width=0.5mm] (A) to[bend left=\angle] (B);
    %   \draw[line width=0.5mm] (A) to[bend right=\angle] (B);
    % }
    
    % % --- (B--C): two evenly spaced arcs ---
    % \foreach \angle in {15} {
    %   \draw[line width=0.5mm] (B) to[bend left=\angle] (C);
    %   \draw[line width=0.5mm] (B) to[bend right=\angle] (C);
    % }
    
    % % --- (C--D): two evenly spaced arcs ---
    % \foreach \angle in {15} {
    %   \draw[line width=0.5mm] (C) to[bend left=\angle] (D);
    %   \draw[line width=0.5mm] (C) to[bend right=\angle] (D);
    % }
    
    % % --- (D--A): two evenly spaced arcs ---
    % \foreach \angle in {15} {
    %   \draw[line width=0.5mm] (D) to[bend left=\angle] (A);
    %   \draw[line width=0.5mm] (D) to[bend right=\angle] (A);
    % }
    
    % % --- (B--D): two evenly spaced arcs ---
    % \foreach \angle in {15} {
    %   \draw[line width=0.5mm] (B) to[bend left=\angle] (D);
    %   \draw[line width=0.5mm] (B) to[bend right=\angle] (D);
    % }

% \end{scope}

\begin{scope}[xshift = 8cm]

    % Define nodes
    \node[circle, draw=black, fill=black] (A) at (-1, -1){};
    \node[circle, draw=black, fill=black] (B) at (-1, 1){};
    \node[circle, draw=black, fill=black] (C) at (1, 1){};
    \node[circle, draw=black, fill=black] (D) at (1, -1){};
    
    % --- (A--B): four evenly spaced arcs ---
    \foreach \angle in {5, 15} {
      \draw[line width=0.5mm] (A) to[bend left=\angle] (B);
      \draw[line width=0.5mm] (A) to[bend right=\angle] (B);
    }
    
    % --- (B--C): four evenly spaced arcs ---
    \foreach \angle in {5, 15} {
      \draw[line width=0.5mm] (B) to[bend left=\angle] (C);
      \draw[line width=0.5mm] (B) to[bend right=\angle] (C);
    }
    
    % --- (C--D): four evenly spaced arcs ---
    \foreach \angle in {5, 15} {
      \draw[line width=0.5mm] (C) to[bend left=\angle] (D);
      \draw[line width=0.5mm] (C) to[bend right=\angle] (D);
    }
    
    % --- (D--A): four evenly spaced arcs ---
    \foreach \angle in {5, 15} {
      \draw[line width=0.5mm] (D) to[bend left=\angle] (A);
      \draw[line width=0.5mm] (D) to[bend right=\angle] (A);
    }
    
    % --- (B--D): four evenly spaced arcs ---
    \foreach \angle in {5, 15} {
      \draw[line width=0.5mm] (B) to[bend left=\angle] (D);
      \draw[line width=0.5mm] (B) to[bend right=\angle] (D);
    }

\end{scope}

\end{tikzpicture}}
	\caption{A graph $G$, the thickening gadget $\mathbf{T_{4}}$, and $\mathbf{T_{4}}G$.}
	\label{fig:thickening}
\end{figure}

From \cref{equation: signature}, we have that for
all $i, j \in [q]$,

$$\mathbf{T_{n}}(M)_{ij} = 
\sum_{\substack{\tau: V(\mathbf{T_{n}}) \rightarrow [q]\\
\tau(\ell_{1}) = i, \tau(\ell_{2}) = j}}
\prod_{(u, v) \in E(\mathbf{T_{n}})}M_{\tau(u) \tau(v)} 
= \prod_{(u, v) \in E(\mathbf{T_{n}})}M_{ij}
= (M_{ij})^{n}.$$

% $$Z_{M}(\mathbf{T_{n}}G)
% = \sum_{\sigma: V \rightarrow [q]}
% \prod_{(u, v) \in \mathbf{T_{n}}E}M_{\sigma(u)\sigma(v)}
% = \sum_{\sigma: V \rightarrow [q]}
% \prod_{(u, v) \in E}(M_{\sigma(u)\sigma(v)})^{n}
% = Z_{\mathbf{T_{n}}(M)}(G),$$
% where $\mathbf{T_{n}}M = \big((M_{ij})^{n}\big) \in \Sym{q}{}(\mathbb{R})$.
% %with entries $\big((M_{ij})^{n}\big)$ for $i, j \in [q]$.
% It follows that $\PlGH(\mathbf{T_{n}}M) \le_p^T \PlGH(M)$ for all $n \geq 1$.

We see $\mathbf{T_{n}}(M) \in \PlEdge(M)$ for all
$M \in \Sym{q}{}{\mathbb{R}}$ and $n \geq 1$.
So, $\PlGH(\mathbf{T_{n}}(M)) \leq^{T}_{p} \PlGH(M)$
for all $n \geq 1$.
By a technique called \emph{gadget interpolation}
we can reduce from more matrices than $\mathbf{T_{n}}(M)$,
replacing the integer $n$ by a real number.
The idea 
is that with
oracle access to $\PlGH(M)$ we can compute
$Z_{M}(\mathbf{T_{n}}G) = Z_{\mathbf{T_{n}}(M)}(G)$ 
for 
polynomially many $n$, and then form a Vandermonde system
of linear equations. 
By solving this system, we can compute 
$Z_{N}(G)$ for much more general matrices $N$
(see Appendix \ref{appendix: thickeningAppendix}).
A consequence is the following:

\begin{definition}\label{definition: mathcalTSmall}
    Let $M \in \Sym{q}{}{\mathbb{R}_{> 0}}$.
    Define
    $\mathcal{T}_{M}: \mathbb{R} \rightarrow
    \Sym{q}{}{\mathbb{R}_{> 0}}$ such that
    $\mathcal{T}_{M}(\theta)_{ij} = (M_{ij})^{\theta},$ for all  $i, j \in [q].$
\end{definition}

\begin{restatable}{lemma}{simpleThickening}
\label{lemma: simpleThickening}
    Let $M \in \Sym{q}{}{\mathbb{R}_{> 0}}$. Then,
    $\PlGH(\mathcal{T}_{M}(\theta)) \leq^{T}_{p} 
    \PlGH(M)$ for all $\theta \in \mathbb{R}$.
\end{restatable}

The full power of the thickening gadget
allows us to interpolate even more matrices.

\begin{restatable}{definition}{generatingSet}
\label{definition: generatingSet}
    Let $\mathcal{A} \subseteq \mathbb{R}_{\neq 0}$
    be a set of non-zero real numbers.
    A finite set 
    $\{g_{t}\}_{t \in [d]} \in (\mathbb{R}_{> 1})^{d}$,
    for some integer $d \geq 0$, 
    is called a generating set of $\mathcal{A}$ if
    %(1) $\widehat{\mathfrak{L}}(g_1,g_2,\ldots, g_d)=\{\mathbf{0}\}$;(2) 
    for every $a \in \mathcal{A}$, there exists a
    \emph{unique} $(e_0, e_{1}, \dots, e_{d}) \in \{0, 1\} \times {\mathbb{Z}}^{d}$ 
    such that $a = (-1)^{e_0} {g_{1}^{e_{1}} \cdots g_{d}^{e_{d}}}$.
\end{restatable}

\begin{remark*}
    %Note that g
    Given any finite $\mathcal{A} \subset 
    \mathbb{Z}_{\neq 0}$,
    the set of prime factors of 
    $\prod_{a \in \mathcal{A}}|a|$
    forms a generating set of $\mathcal{A}$.
    In general, every finite 
    $\mathcal{A} \subset \mathbb{R}_{\neq 0}$ 
    has a generating set (see \cref{lemma: generatingSet} in
    Appendix \ref{appendix: thickeningAppendix}).
\end{remark*}
%This  allows us to define a more general
%$\mathcal{T}^{*}_{M}$.

\begin{definition}\label{definition: mathcalTLarge}
    Let $M \in \Sym{q}{}{\mathbb{R}_{\neq 0}}$, with a generating set $\{g_t\}_{t \in [d]}$  for its entries,
    %with its entries generated by some $\{g_t\}_{t \in [d]}$,
    %such that 
    $M_{ij} = (-1)^{e_{ij0}} \cdot
    g_{1}^{e_{ij1}} \cdots g_{d}^{e_{ijd}}$.
    Let $e_{t}^{*} = \min_{i, j \in [q]}e_{ijt}$
    for all $t \in [d]$.
    %We may now d
    Define $\mathcal{T}(M,\cdot): \mathbb{R}^{d} \rightarrow
    \Sym{q}{}{\mathbb{R}}$
    such that
    $\mathcal{T}(M;\mathbf{z})_{ij} = 
    (-1)^{e_{ij0}} \cdot z_{1}^{e_{ij1} - e_{1}^{*}} \cdots
    z_{d}^{e_{ijd} - e_{d}^{*}}$
    is a signed monomial in $\mathbf{z} = (z_{1}, \dots, z_{d})$
    for all $i, j \in [q]$.
\end{definition}

\begin{restatable}{lemma}{thickeningLemma}\label{lemma: thickeningLemma}
    Let $M \in \Sym{q}{}{\mathbb{R}_{\neq 0}}$, with a generating set $\{g_t\}_{t \in [d]}$  for its entries.
    %with its entries generated by some $\{g_t\}_{t \in [d]}$.
    Then, $\PlGH(\mathcal{T}(M;\mathbf{z}))$ $ \leq^{T}_{p} \PlGH(M)$
    for all $\mathbf{z} \in \mathbb{R}^{d}$.
\end{restatable}

%%% JYC to save space I am commenting it out...
%The gadget function $\mathcal{T}^{*}_{M}$ can be unwieldy, so we use the following to make it easier to work with.

A proof as well as a detailed discussion of \cref{lemma: thickeningLemma} can be found in Appendix \ref{appendix: thickeningAppendix}.
\subsection{Stretching Gadgets}\label{sec: stretchingGadgets}

For an integer $n \geq 1$, the \emph{stretching gadget}
$\mathbf{S_{n}}$ consists of two vertices
$\ell_{1}, \ell_{2}$, that are connected
by a path of length $n$
(see \cref{fig:stretching}).
\begin{figure}[htbp]
	\centering
	\scalebox{0.8}{\begin{tikzpicture}[line join=miter, draw opacity=1]

    % Define nodes
    \node[circle, draw=black, fill=black] (A) at (-1, -1){};
    \node[circle, draw=black, fill=black] (B) at (-1, 1){};
    \node[circle, draw=black, fill=black] (C) at (1, 1){};
    \node[circle, draw=black, fill=black] (D) at (1, -1){};
    
    % --- (A--B): one edge ---
    \draw[line width=0.5mm, black] (A) -- (B);
    
    % --- (B--C): one edge ---
    \draw[line width=0.5mm, black] (B) -- (C);
    
    % --- (C--D): one edge ---
    \draw[line width=0.5mm, black] (C) -- (D);
    
    % --- (D--A): one edge ---
    \draw[line width=0.5mm, black] (D) -- (A);
    
    % --- (B--D): one edge ---
    \draw[line width=0.5mm, black] (B) -- (D);

\begin{scope}[xshift = 4cm]

    % Define nodes
    \node[circle, draw=black, fill=white] (A) at (-1, 0){};
    \node[circle, draw=black, fill=white] (B) at (1, 0){};

    % Label nodes
    \node[draw=none] at (-0.75, -0.5) {$\ell_{1}$};
    \node[draw=none] at (0.75, -0.5) {$\ell_{2}$};
    
    % Define sub-nodes A -- B
    \node[circle, draw=black, fill=black, inner sep=0pt, minimum size=8] (AB1) at (-0.5, 0){};
    \node[circle, draw=black, fill=black, inner sep=0pt, minimum size=8] (AB2) at (0, 0){};
    \node[circle, draw=black, fill=black, inner sep=0pt, minimum size=8] (AB3) at (0.5, 0){};
    
    % --- (A--B): two path of length 4 ---
    \draw[line width=0.5mm, black] (A) -- (AB1);
    \draw[line width=0.5mm, black] (AB1) -- (AB2);
    \draw[line width=0.5mm, black] (AB2) -- (AB3);
    \draw[line width=0.5mm, black] (AB3) -- (B);
    
\end{scope}

% \begin{scope}[xshift = 4cm]

    % % Define nodes
    % \node[circle, draw=black, fill=black] (A) at (-1, -1){};
    % \node[circle, draw=black, fill=black] (B) at (-1, 1){};
    % \node[circle, draw=black, fill=black] (C) at (1, 1){};
    % \node[circle, draw=black, fill=black] (D) at (1, -1){};
    
    % % Define sub-nodes A -- B
    % \node[circle, draw=black, fill=black, inner sep=0pt, minimum size=8] (AB1) at (-1, 0){};
    
    % % Define sub-nodes B -- C
    % \node[circle, draw=black, fill=black, inner sep=0pt, minimum size=8] (BC1) at (0, 1){};
    
    % % Define sub-nodes C -- D
    % \node[circle, draw=black, fill=black, inner sep=0pt, minimum size=8] (CD1) at (1, 0){};
    
    % % Define sub-nodes D -- A
    % \node[circle, draw=black, fill=black, inner sep=0pt, minimum size=8] (DA1) at (0, -1){};
    
    % % Define sub-nodes B -- D
    % \node[circle, draw=black, fill=black, inner sep=0pt, minimum size=8] (BD1) at (0, 0){};
    
    % % --- (A--B): two path of length 2 ---
    % \draw[line width=0.5mm, black] (A) -- (AB1);
    % \draw[line width=0.5mm, black] (AB1) -- (B);
    
    % % --- (B--C): two path of length 2 ---
    % \draw[line width=0.5mm, black] (B) -- (BC1);
    % \draw[line width=0.5mm, black] (BC1) -- (C);
    
    % % --- (C--D): two path of length 2 ---
    % \draw[line width=0.5mm, black] (C) -- (CD1);
    % \draw[line width=0.5mm, black] (CD1) -- (D);
    
    % % --- (D--A): two path of length 2 ---
    % \draw[line width=0.5mm, black] (D) -- (DA1);
    % \draw[line width=0.5mm, black] (DA1) -- (A);
    
    % % --- (B--D): two path of length 2 ---
    % \draw[line width=0.5mm, black] (B) -- (BD1);
    % \draw[line width=0.5mm, black] (BD1) -- (D);

% \end{scope}

\begin{scope}[xshift = 8cm]

    % Define nodes
    \node[circle, draw=black, fill=black] (A) at (-1, -1){};
    \node[circle, draw=black, fill=black] (B) at (-1, 1){};
    \node[circle, draw=black, fill=black] (C) at (1, 1){};
    \node[circle, draw=black, fill=black] (D) at (1, -1){};
    
    % Define sub-nodes A -- B
    \node[circle, draw=black, fill=black, inner sep=0pt, minimum size=8] (AB1) at (-1, -0.5){};
    \node[circle, draw=black, fill=black, inner sep=0pt, minimum size=8] (AB2) at (-1, 0){};
    \node[circle, draw=black, fill=black, inner sep=0pt, minimum size=8] (AB3) at (-1, 0.5){};
    
    % Define sub-nodes B -- C
    \node[circle, draw=black, fill=black, inner sep=0pt, minimum size=8] (BC1) at (-0.5, 1){};
    \node[circle, draw=black, fill=black, inner sep=0pt, minimum size=8] (BC2) at (0, 1){};
    \node[circle, draw=black, fill=black, inner sep=0pt, minimum size=8] (BC3) at (0.5, 1){};
    
    % Define sub-nodes C -- D
    \node[circle, draw=black, fill=black, inner sep=0pt, minimum size=8] (CD1) at (1, 0.5){};
    \node[circle, draw=black, fill=black, inner sep=0pt, minimum size=8] (CD2) at (1, 0){};
    \node[circle, draw=black, fill=black, inner sep=0pt, minimum size=8] (CD3) at (1, -0.5){};
    
    % Define sub-nodes D -- A
    \node[circle, draw=black, fill=black, inner sep=0pt, minimum size=8] (DA1) at (0.5, -1){};
    \node[circle, draw=black, fill=black, inner sep=0pt, minimum size=8] (DA2) at (0, -1){};
    \node[circle, draw=black, fill=black, inner sep=0pt, minimum size=8] (DA3) at (-0.5, -1){};
    
    % Define sub-nodes B -- D
    \node[circle, draw=black, fill=black, inner sep=0pt, minimum size=8] (BD1) at (-0.5, 0.5){};
    \node[circle, draw=black, fill=black, inner sep=0pt, minimum size=8] (BD2) at (0, 0){};
    \node[circle, draw=black, fill=black, inner sep=0pt, minimum size=8] (BD3) at (0.5, -0.5){};

    % --- (A--B): two path of length 4 ---
    \draw[line width=0.5mm, black] (A) -- (AB1);
    \draw[line width=0.5mm, black] (AB1) -- (AB2);
    \draw[line width=0.5mm, black] (AB2) -- (AB3);
    \draw[line width=0.5mm, black] (AB3) -- (B);
    
    % --- (B--C): two path of length 4 ---
    \draw[line width=0.5mm, black] (B) -- (BC1);
    \draw[line width=0.5mm, black] (BC1) -- (BC2);
    \draw[line width=0.5mm, black] (BC2) -- (BC3);
    \draw[line width=0.5mm, black] (BC3) -- (C);
    
    % --- (C--D): two path of length 4 ---
    \draw[line width=0.5mm, black] (C) -- (CD1);
    \draw[line width=0.5mm, black] (CD1) -- (CD2);
    \draw[line width=0.5mm, black] (CD2) -- (CD3);
    \draw[line width=0.5mm, black] (CD3) -- (D);
    
    % --- (D--A): two path of length 4 ---
    \draw[line width=0.5mm, black] (D) -- (DA1);
    \draw[line width=0.5mm, black] (DA1) -- (DA2);
    \draw[line width=0.5mm, black] (DA2) -- (DA3);
    \draw[line width=0.5mm, black] (DA3) -- (A);
    
    % --- (B--D): two path of length 4 ---
    \draw[line width=0.5mm, black] (B) -- (BD1);
    \draw[line width=0.5mm, black] (BD1) -- (BD2);
    \draw[line width=0.5mm, black] (BD2) -- (BD3);
    \draw[line width=0.5mm, black] (BD3) -- (D);

\end{scope}

\end{tikzpicture}}
	\caption{A graph $G$, the stretching gadget $\mathbf{S_{4}}$, and $\mathbf{S_{4}}G$.}
	\label{fig:stretching}
\end{figure}
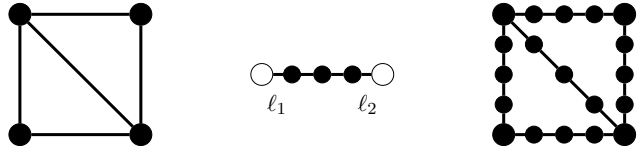

From \cref{equation: signature}, we have that for all
$i, j \in [q]$,
$$\mathbf{S_{n}}(M)_{ij} = 
\sum_{\substack{\tau: V(\mathbf{S_{n}}) \rightarrow [q]\\
\tau(\ell_{1}) = i, \tau(\ell_{2}) = j}}
\prod_{(u, v) \in E(\mathbf{S_{n}})}M_{ij}
= \sum_{k_{1}, \dots, k_{n-1} \in [q]}
M_{ik_{1}}M_{k_{1}k_{2}} \cdots M_{k_{n-1}j}
= (M^{n})_{ij}.$$

% Given any planar multi-graph $G = (V, E),$ and an
% integer $n \ge 1$, we replace every edge in $G$ by
% a  path of length $n$ to obtain $S_{n}G = (S_{n}V, S_{n}E)$. We have
% \begin{align*}
%     Z_{M}(S_{n}G) 
%     &= \sum_{\sigma: V \rightarrow [q]}
%     \prod_{(u, v) \in E} \left(
%     \sum_{\tau: [n-1] \rightarrow [q]}
%     M_{\sigma(u) \tau(1)} M_{\tau(1)\tau(2)} \cdots
%     M_{\tau(n-2) \tau(n-1)} M_{\tau(n-1) \sigma(v)}\right)
%     = Z_{S_{n}M}(G),
% \end{align*}
% where $S_{n}M = M^{n} \in \Sym{q}{}(\mathbb{R})$.
% It follows that $\PlGH(S_{n}M) \le_p^T \PlGH(M)$ for all $n \geq 1$.

We see that $\mathbf{S_{n}}(M) \in \PlEdge(M)$ for all
$M \in \Sym{q}{}{\mathbb{R}}$, for all $n \geq 1$.
Let $M \in \Sym{q}{pd}{\mathbb{R}}$
with eigenvalues $\lambda_{1}, \dots, \lambda_{q} > 0$.
There exists an 
%orthonormal 
%JYC: just called orthogonal matrix , including the requiement that columns are unit vectors.
orthogonal matrix 
of unit eigenvectors $H$
such that $M = HDH^{\tt{T}}$, where
$D = \diag(\lambda_{1}, \dots, \lambda_{q})$
is a diagonal matrix.
Note that $M^{n} = HD^{n}H^{\tt{T}}$ for
all $n \geq 1$.
So, 
with oracle access to $\PlGH(M)$, we can compute
$Z_{M}(\mathbf{S_{n}}G) = Z_{\mathbf{S_{n}}(M)}(G)$
for polynomially many $n$.
From this, we obtain a Vandermonde system
and 
this lets us compute $Z_{N}(G)$ for other
matrices $N$ (see Appendix \ref{appendix: stretchingAppendix}).
A consequence is the following:

\begin{restatable}{definition}{DefMathcalS}\label{definition: mathcalS}
    Let $M = HDH^{\tt{T}} \in \Sym{q}{pd}{\mathbb{R}}$, where $H\in \mathrm{SO}_q(\R)$, and
    $D = \diag(\lambda_{1}, \dots, \lambda_{q})$.
    Define $\mathcal{S}_{M}: \mathbb{R} \rightarrow 
    \Sym{q}{pd}{\mathbb{R}}$ such that
    $\mathcal{S}_{M}(\theta) = HD^{\theta}H^{\tt{T}},\mbox{~for all real $\theta \in \mathbb{R}$},$
    where $D^{\theta} = \diag(\lambda_{1}^{\theta},
    \dots, \lambda_{q}^{\theta})$.
\end{restatable}

It can be shown that $\mathcal{S}_M$ is independent of the choice of the decomposition $HDH^{\tt{T}}$ (see \cref{lemma: real power of pd matrix} in Appendix \ref{appendix: stretchingAppendix}).

\begin{restatable}{lemma}{stretchingLemma}\label{lemma: stretchingLemma}
    Let $M \in \Sym{q}{pd}{\mathbb{R}}$.
    Then $\PlGH(\mathcal{S}_{M}(\theta)) \leq^{T}_{p} \PlGH(M)$
    for all $\theta \in \mathbb{R}$.
\end{restatable}

The full power of the stretching gadget
allows us to interpolate even more matrices.
%(see Appendix \ref{appendix: stretchingAppendix}).
\begin{restatable}{definition}{latticeSet}\label{definition: latticeSet}
    Let $(a_1,a_2,\ldots,a_{q})\in \R_{\neq 0}^q$.
    Define
    
    $$\widehat{\mathfrak{L}}(a_1,\ldots,a_q)
    =\left\{\mathbf{x}\in\Z^q \mid \prod_{i=1}^q a_i^{x_i}=1\right\},
    \qquad
    \mathfrak{L}(a_1,\ldots,a_q)
    =\left\{\mathbf{x}\in\widehat{\mathfrak{L}}(a_1,\ldots,a_q) \mid \sum_{i=1}^q x_i=0\right\}.$$ 
    
\end{restatable}

\begin{restatable}{lemma}{latticeInterpolation}\label{lemma: latticeInterpolation}
    Let $M \in \text{Sym}^{\tt{F}}_{q}(\mathbb{R})$,
    such that $M = HDH^{\tt{T}}$, where
    $D = \text{\rm diag}(\lambda_{1}, \dots, \lambda_{q})$.
    Then $\PlGH(H\Delta H^{\tt{T}}) \leq^{T}_{p} \PlGH(M)$ 
    for any diagonal matrix
    $\Delta = \text{diag}(\Delta_{1}, \dots, \Delta_{q})$,
    such that
    $\Delta_{i} \in \mathbb{R}_{\neq 0}$ for all $i \in [q]$,
    and $\mathfrak{L}(\lambda_{1}, \dots, \lambda_{q})
    \subseteq \mathfrak{L}(\Delta_{1}, \dots, \Delta_{q})$.
\end{restatable}

\cref{lemma: latticeInterpolation} is first proved in \cite{cai2023complexity}.
We give an alternative proof based on the ``conformal lattice interpolation'' proposed in \cite{CaiFST26} in Appendix \ref{appendix: stretchingAppendix}.

\subsection{Hardness Source}
\begin{lemma}[\cite{vertigan2005computational}]\label{lemma: Potts Model}
    The $q$-state Potts model
     $\PlGH({\tt{Potts}}_{q}(x))$ is $\#$P-hard for any integer $q \geq 3$,
    and non-negative real $x \neq 1$,
    where ${\tt{Potts}}_{q}(x) \in \text{Sym}_{q}(\mathbb{R})$
    is the matrix with entries
    $({\tt{Potts}}_{q}(x)_{ij})_{i,j\in[q]}$ such that
    ${\tt{Potts}}_{q}(x)_{ij} = 1$ if $i \neq j$,
    and ${\tt{Potts}}_{q}(x)_{ij} = x$ otherwise.
\end{lemma}

\begin{definition}\label{def: lattice of potts}
    Define 
    $\LPotts{q}=\{(x_1,x_2,\ldots,x_q)\in \Z^q\mid \sum_{i\in [q]}x_i=0 \text{ and }x_1=0\}.$
\end{definition}

Note that $\PlGH({\tt{Potts}}_{q}(0))$ is the problem of counting vertex coloring with $q$ colors on planar graphs.
And $\LPotts{q}$ is the spectral lattice of the Potts model (see \cref{rmk: lattice of Potts} in \autoref{appendix: Circulant}).

\section{Hardness Criteria}\label{sec: I+uuT}
In this section, we prove two hardness criteria for $\PlGH$ problems, the spectral criterion (\cref{thm: only look at lattice of eigenvalues}) and the determinant criterion (\cref{thm: determinant criterion}), by reducing from the Potts model.
These criteria hold even when the matrix defining the $\PlGH$ problem has quantum symmetry, i.e., having a nontrivial $\qut(M)$.
\cref{thm: only look at lattice of eigenvalues} is established by showing that the $\PlGH$ problem defined by a rank one perturbation of the identity matrix is \#P-hard (\cref{lm: I+uu^T is hard}).
We establish \cref{lm: I+uu^T is hard} by performing stretching (\cref{lm: scalar on u in I+uuT by stretching}) to eliminate multiplicative relations among diagonals and off-diagonals of $M$ (\cref{lm: generating set of I+uuT}). 
Then, we appeal to the thickening gadgets and generating sets to reduce from the $q$-vertex coloring problem $\PlGH(\potts{q}{0})$.
\cref{thm: determinant criterion} is an application of \cref{thm: only look at lattice of eigenvalues} coupled with results in matrix perturbation theory.

\subsection{Spectral Criterion}

\begin{lemma}\label{lm: scalar on u in I+uuT by stretching}
    Let $\mathbf{u}=(u_1,u_2,\ldots,u_q)^{\tt T}\in \Rp^q$ with $||\mathbf{u}||_2=1$ and $I=I_q$.
    Then
    $\PlGH(cI+\mathbf{u}\mathbf{u}^{\tt{T}})\le_p^T \PlGH(I+\mathbf{u}\mathbf{u}^{\tt{T}})$ 
    for any $c\in\Rp.$
\end{lemma}
\begin{proof}
    Note that $\mathbf{u}\mathbf{u}^{\tt{T}}$ is a rank-1 matrix. It possesses a single non-zero eigenvalue 1 with the associated eigenvector $\mathbf{v}_1 = {\mathbf{u}}$. 
    All other $q-1$ eigenvalues are 0, with corresponding orthonormal eigenvectors $\mathbf{v}_i$ for $2 \le i \le q$.
    So $I+\mathbf{u}\mathbf{u}^{\tt{T}}=H\mathrm{diag}(2,1,\ldots,1)H^{\tt{T}},$
    where $H=(\mathbf{v}_1,\mathbf{v}_2,\ldots,\mathbf{v}_q),$
    and 
   
    \begin{equation}
        \begin{aligned}
            \mathcal{S}_{I+\mathbf{u}\mathbf{u}^{\tt{T}}}(\theta)=&H\mathrm{diag}(2^\theta,1,\ldots,1)H^{\tt{T}}\\
            =&I+H\mathrm{diag}(2^\theta-1,0,\ldots,0)H^{\tt{T}}\\
            =&I+(2^\theta-1)\mathbf{u}\mathbf{u}^{\tt{T}},~~\theta\in \Rp.
        \end{aligned}
    \end{equation}
    
    Let $c(\theta):=2^\theta-1$, $\theta\in \Rp$, then $\mathcal{S}_{I+\mathbf{u}\mathbf{u}^{\tt{T}}}(\theta)=I+c(\theta)\mathbf{u}\mathbf{u}^{\tt{T}}.$
    Notice that $c(\theta)$ is an increasing function on $(0,\infty)$, and its range is $(0,\infty).$
    Since $\PlGH(\mathcal{S}_{I+\mathbf{u}\mathbf{u}^{\tt{T}}}(\theta))\le_p^T\PlGH(I+\mathbf{u}\mathbf{u}^{\tt{T}})$ for any $\theta\in\Rp$ by \cref{lemma: stretchingLemma}, we have $\PlGH(I+c\mathbf{u}\mathbf{u}^{\tt{T}})\le_p^T \PlGH(I+\mathbf{u}\mathbf{u}^{\tt{T}})$
    for any $c\in\Rp.$
    Then, by \cref{lemma: MequivalentCM}, $\PlGH(cI+\mathbf{u}\mathbf{u}^{\tt{T}}) \equiv_p^T \PlGH(I+ c^{-1}\mathbf{u}\mathbf{u}^{\tt{T}}) \le_p^T \PlGH(I+\mathbf{u}\mathbf{u}^{\tt{T}})$
    for any $c\in\Rp$.
\end{proof}

\vspace{-0.1in}

\begin{restatable}{lemma}{GeneratingSetofRankOnePerturbation}\label{lm: generating set of I+uuT}
    Let $q\ge 3$, $\mathbf{u}=(u_1,u_2,\ldots,u_q)^{\tt T}\in \Rp^q$ and $I=I_q$.
    There exists some $c\in \Rp$, such that there are \emph{disjoint} generating sets of diagonals and off-diagonals of $cI+\mathbf{uu}^{\tt T}$ respectively.
    Moreover, their union is a generating set of all entries of $cI+\mathbf{uu}^{\tt T}$.
\end{restatable}

\begin{lemma}\label{lm: I+uu^T is hard}
    Let $q\ge 3$, $\mathbf{u}
    %=(u_1,u_2,\ldots,u_q)^{\tt T}
    \in \Rp^q$ with $||\mathbf{u}||_2=1$.
    %and $I=I_q$.
    Then $\PlGH(I+\mathbf{u}\mathbf{u}^{\tt{T}})$ is \numPhard.
\end{lemma}
\begin{proof}
    %By \cref{lemma: MequivalentCM,lm:  scalar on u in I+uuT by stretching}, we may normalize $||\mathbf{u}||_2=1.$
    By \cref{lm: generating set of I+uuT},
    we may pick $c\in \Rp$, 
    $\{g_t\mid 1\le t\le d_1\}$ and $\{g_t\mid d_1+1\le t\le d\}$ be generating sets of diagonals and off-diagonals of $cI+\mathbf{uu}^{\tt T}$ respectively, so that $\{g_t\mid t\in[d]\}$ is a generating set of all entries of $cI+\mathbf{uu}^{\tt T}.$
    By \cref{lemma: thickeningLemma}, $\PlGH(\mathcal{T}(cI+\mathbf{uu}^{\tt T};\mathbf{z}))\le_p^T \PlGH(cI+\mathbf{uu}^{\tt T})$ for $\mathbf{z}=\mathbf{0}^{d_1}\mathbf{1}^{d-d_1}.$
    Notice that $\mathcal{T}(cI+\mathbf{uu}^{\tt T};\mathbf{z})=\potts{q}{0}$, which defines the $q$-vertex coloring problem. 
    So $\PlGH(cI+\mathbf{uu}^{\tt T})$ is \#P-hard.
    By \cref{lm: scalar on u in I+uuT by stretching}, $\PlGH(I+\mathbf{uu}^{\tt T})$ is \#P-hard.
\end{proof}

\begin{theorem}\label{thm: only look at lattice of eigenvalues}
    Suppose $q\ge 3.$
    Let $M\in \Sym{q}{F}{\Rp}$, and $\lambda_1,\lambda_2,\ldots,\lambda_q$ be the eigenvalues of $M$, where $\lambda_1>|\lambda_i|,\,\forall 2\le i\le q.$
    If $\mathfrak{L}(\lambda_1,\lambda_2,\ldots,\lambda_q)\subseteq \LPotts{q},$ then $\PlGH(M)$ is \numPhard.
\end{theorem}

\begin{proof}
    The existence of $\lambda_1$ is guaranteed by the Perron-Frobenius theorem.
    Let $M=HDH^{\tt{T}}$, where $D=\mathrm{diag}(\lambda_1,\lambda_2,\ldots,\lambda_q)$ and $H\in \mathrm{SO}_q(\mathbb{R})$.
    By assumption we know
    $
    \mathfrak{L}(\lambda_1,\lambda_2,\ldots,\lambda_q)\subseteq \LPotts{q}=\mathfrak{L}(2,1,\ldots,1).
    $
    By \cref{lemma: latticeInterpolation},
    $\PlGH(N)\le_p^T \PlGH(M)$,
    where $N=H\mathrm{diag}(2,1,\ldots,1)H^{\tt{T}}$.
    Note that $N=I+H\mathrm{diag}(1,0,\ldots,0)H^{\tt{T}}=I+\mathbf{u}\mathbf{u}^{\tt{T}}$, where $\mathbf{u}\in \Rp^q$ is the norm 1 Perron eigenvector corresponding to $\lambda_1$. 
    By \cref{lm: I+uu^T is hard}, $\PlGH(I+\mathbf{u}\mathbf{u}^{\tt{T}})$ is \numPhard.
    Thus $\PlGH(M)$ is \numPhard.
\end{proof}

\subsection{Determinant Criterion}

By \cref{thm: only look at lattice of eigenvalues}, if we can realize a matrix $M'\in \Sym{q}{F}{\Rp}$ from $M\in \Sym{q}{}{\Rp}$ whose spectral lattice is contained in $\LPotts{q}$, then $\PlGH(M)$ is \#P-hard.
For example, we may consider $M'=\mathcal{T}_M(\theta)$ as $\theta\to 0^+.$
This is a small perturbation of the all one matrix $J_q=\mathcal{T}_M(0).$
So the eigenvalues $(\lambda_1(\theta),\ldots, \lambda_q(\theta))$ of $M'$ are small perturbations of $(q,0,\ldots, 0)$, which are eigenvalues of $J_q.$
By the following result in matrix perturbation theory (see for example \cite{kato1995perturbation} Chapter II, Theorem 6.1), $\lambda_i(\theta)$'s are analytic functions of $\theta$:

\begin{lemma}[Rellich]\label{lm: Rellich}
    Let $M(\theta):\R \to \Sym{q}{}{\R}$ be a matrix function that depends on $\theta$ entry-wise real analytically.
    Then the $q$ roots of the characteristic polynomial of $M(\theta)$ can be arranged so that each root $\lambda_j(\theta)$ for $1\le j\le q$ is a real analytic function of $\theta$.
\end{lemma}

\begin{lemma}\label{lm: none-zero linear coe of eigenvalues implies hardness}
    Suppose $q\ge 3.$
    Let $M\in \Sym{q}{}{\Rp}$.
    %Consider $\mathcal{T}_M(\theta)$ as $\theta\to 0^+.$
    Let $\{\lambda_j(\theta)\mid 1\le j\le q\}$ be the eigenvalues of $\mathcal{T}_M(\theta).$ 
    We may pick $\lambda_1(\theta)$ to be the unique largest eigenvalue in a small neighborhood $(0,\epsilon)$.
    Moreover,
    if $\lambda_j(\theta)=\Theta(\theta)$ as $\theta\to 0^+$ for every $2\le j\le q$, then $\PlGH(M)$ is \#P-hard.
\end{lemma}

\begin{proof}
    By \cref{lm: Rellich}, we may assume $\lambda_1(\theta)=q+\sum_{n=1}^{+\infty}b_n^{(1)}\theta^n$ and $\lambda_j(\theta)=\sum_{n=1}^{+\infty}b_n^{(j)}\theta^n\,(\forall 2\le j\le q)$ where the coefficients $b_n^{(1)},b_n^{(j)}$ are all real. 
    By assumption, $b_1^{(j)}\neq 0$ for every $2\le j\le q.$
    So there exists $\epsilon>0$ such that $\lambda_j(\theta)\neq 0,\forall 1\le j\le q,\forall\theta\in (0,\epsilon).$
    For any $\mathbf{x}=(x_1,x_2,\ldots,x_q)\in \Z^q\setminus \{\mathbf{0}\}$ such that $\sum_{j=1}^q x_j=0$, define $f_\mathbf{x}(\theta):=\sum_{j=1}^q x_j \ln |\lambda_j(\theta)|$ for $\theta\in (0,\epsilon).$
    The condition $\mathbf{x}\in \mathfrak{L}(\lambda_1(\theta),\ldots,\lambda_q(\theta))$ implies $f_{\mathbf{x}}(\theta)=0$. 
    We have
 
    \begin{equation}\label{eq: lattice condition of pertubation of J}
        \begin{aligned}
            &f_{\mathbf{x}}(\theta)= 
            x_1\ln |\lambda_1(\theta)|+\sum_{j=2}^{q}x_j\ln |\lambda_j(\theta)|= x_1\ln\left|q+\sum_{n=1}^{+\infty}b_n^{(1)}\theta^n\right|+\sum_{j=2}^q x_j \ln\left|\sum_{n=1}^{+\infty}b_n^{(j)}\theta^n\right|\\
            =&x_1\ln q+x_1\ln\left|1+\sum_{n=1}^{+\infty}\frac{b_n^{(1)}}{     q}\theta^n\right|+\sum_{j=2}^qx_j\ln |b_1^{(j)}\theta| +\sum_{j=2}^q x_j \ln\left|1+\sum_{n=2}^{+\infty}\frac{b_n^{(j)}}{|b_1^{(j)}|}\theta^{n-1}\right|\\
            =&
            -\sum_{j=2}^q x_j \ln (1/\theta)+ x_1 \ln q +\sum_{j=2}^q x_j \ln |b_1^{(j)}| + o(1), \qquad \text{ as $\theta\to 0^+$}.
            %%% JYC better to say o(1) where you had O(1)
        \end{aligned}
    \end{equation}
    
    The leading coefficient of $f_{\mathbf{x}}(\theta)$ is $-\sum_{j=2}^q x_j=x_1$.
    Notice that $\mathbf{x}\in \LPotts{q}\iff x_1=0.$
    So if $\mathbf{x}\notin \LPotts{q}$ then $x_1\neq 0$, and $f_{\mathbf{x}}(\theta)$ is a non-zero analytic function.
    Hence the zeros of $f_{\mathbf{x}}(\theta)$ is countable.
    Further, the union of zero sets of $f_{\mathbf{x}}(\theta)$ for all $\mathbf{x}\in \mathbb{Z}^q$ such that $\mathbf{x}\notin \LPotts{q}$ is countable.
    Thus, there exists $\theta\in(0,\epsilon)$ so that $f_{\mathbf{x}}(\theta)\neq 0$ for every $\mathbf{x}\notin \LPotts{q}$. 
    For this $\theta$, $\mathbf{x}\in \mathfrak{L}(\lambda_1(\theta),\ldots,\lambda_q(\theta))\implies f_{\mathbf{x}}(\theta)=0\implies\mathbf{x}\in \LPotts{q}$.
    So
    $\mathfrak{L}(\lambda_1(\theta),\ldots,\lambda_q(\theta))\subseteq \LPotts{q}$. 
    Then $\PlGH(\mathcal{T}_M(\theta))$ is \#P-hard by \cref{thm: only look at lattice of eigenvalues}.
    Thus $\PlGH(M)$ is \#P-hard by \cref{lemma: simpleThickening}.
\end{proof}

Matrix perturbation theory not only tells us the eigenvalues $\lambda_j(\theta)$ are analytic, but also tells us their asymptotic rate as $\theta\to 0^+$:

\begin{lemma}[\cite{kato1995perturbation} Chapter II, Theorem 5.4]\label{lm: Kato}
Under the assumption of \cref{lm: Rellich}, let $\lambda$ be an eigenvalue of $M(0)$ of multiplicity $m$, and $\{\lambda_j({\theta})\mid 1\le j\le m\}$ be the analytic eigenvalue branches of $M(\theta)$ satisfying $\lambda_j(0)=\lambda$.
Then
\(
    \lambda_j(\theta) = \lambda + b_j^{}\theta  + o(\theta),\, j = 1, \dots, m,
\)
where $b_j^{}$ are the eigenvalues of $P M'(0) P$ in the subspace $V = P (\R^q)$ (with possible repetitions), where $P$ is the projection onto the eigen-space of $M(0)$ corresponding to $\lambda$, and $M'(\theta)$ is the entry-wise derivative.
%%% JYC repeated?
\end{lemma}
As a consequence, we are able to prove the following theorem:

\begin{theorem}\label{thm: determinant criterion}
    Suppose $q\ge 3.$
    Let $M\in \Sym{q}{}{\Rp}$.
    Define $L\in \Sym{q}{}{\R}$ with $L_{ij}=\ln M_{ij}$.
    If
     \(
    \det \left[
    \begin{matrix}
    L  & \mathbf{1} \\
    \mathbf{1}^{\tt{T}}  & 0\\
    \end{matrix}\right] \neq 0
    \)
    where $\mathbf{1}\in \mathbb{R}^q$ is the all-one vector,
    then $\PlGH(M)$ is \#P-hard.
\end{theorem}

\begin{proof}
    Consider the eigenvalue $\lambda=0$ of $\mathcal{T}_M(0)=J_q$, with multiplicity $q-1.$
    Let $\lambda_2(\theta),\ldots,\lambda_{q}(\theta)$ be the analytic eigenvalue branches of $\mathcal{T}_M(\theta)$ around $\lambda=0$.
    Write $\lambda_j(\theta)=\sum_{n=1}^{+\infty}b_n^{(j)}\theta^n$ by \cref{lm: Rellich}.
    Notice that the eigen-space of $J_q$ corresponding to $\lambda=0$ is $\mathbf{1}^\perp$, and $\mathcal{T}_M'(0)=L.$
    By \cref{lm: Kato}, $\{b_1^{(j)}\mid 2\le j\le q\}$ are the eigenvalues of $PLP|_{\mathbf{1}^\perp}$ (with possible repetitions).
    Here $P:\mathbb{R}^q\to \mathbf{1}^\perp$ is the projection onto the eigen-space $\mathbf{1}^\perp$.
    So $\lambda_j(\theta)=\Theta(\theta)$ as $\theta\to 0^+$ for every $2\le j\le q$ $\iff$ $b_1^{(2)}b_1^{(3)}\ldots b_1^{(q)}\neq 0 \iff PLP|_{\mathbf{1}^\perp} $ is invertible.
    %By linear algebra, this is equivalent to \(\det \left[\begin{matrix} L  & \mathbf{1} \\\mathbf{1}^{\tt{T}}  & 0\\\end{matrix}\right] \neq 0.\)
    By linear algebra, $PLP|_{\mathbf{1}^\perp}$ is invertible iff: $\forall \mathbf{v} \in \mathbf{1}^\perp$, $PL\mathbf{v} =0 \implies\mathbf{v}=0.$
    i.e., $\forall \mathbf{v} \in \mathbf{1}^\perp, L\mathbf{v}=x\mathbf{1} \text{ for some }x\in \mathbb{R} \implies \mathbf{v}=0$.
    i.e., $\forall \mathbf{v}\in \mathbb{R}^q,x\in \mathbb{R},\left[\begin{matrix} L  & \mathbf{1} \\ \mathbf{1}^{\tt{T}}  & 0\\ \end{matrix}\right]\left[\begin{matrix} \mathbf{v}\\ -x \end{matrix}\right]=0 \implies \mathbf{v}=x=0$.
    That is, \(\det \left[\begin{matrix} L  & \mathbf{1} \\\mathbf{1}^{\tt{T}}  & 0\\\end{matrix}\right] \neq 0.\)
    The theorem follows from \cref{lm: none-zero linear coe of eigenvalues implies hardness}.
\end{proof}
%%% JYC: i think in the case of thm 2, we don't need to quote lm 12.
% just direct differentiation , set \theta at 0, will do.

\section{Circulant Matrix of Prime Order}\label{sec: circulant}

In this section, we prove that the $\PlGH$ problem, defined by any positive definite circulant matrix of {\it prime} order, is \#P-hard (\cref{thm: circulant prime order}).
This is an infinite class of matrices with quantum symmetry, so that it is beyond the reach of the vertex separation method.
We establish \#P-hardness by directly reducing from the Potts model of the same size using \cref{thm: only look at lattice of eigenvalues}.   Analyzing the asymptotic behavior of the eigenvalues under a ``$\theta$-stretching" as $\theta \to + \infty$,
we demonstrate that for a sufficiently large $\theta$, we can realize a matrix whose spectral lattice is contained in $\LPotts{q}.$
%In particular, the proof \cref{lm: lattice of c contained in Potts after thickening}
%only holds for prime $q$, utilizing a property of the cyclotomic extension $\Q(\zeta_q)/\Q$.
%%% JYC I rewrote this sentence
The proof of \cref{lm: lattice of c contained in Potts after thickening}
uses some special property of the cyclotomic extension $\Q(\zeta_q)/\Q$ which is only valid for primes $q \ge 3$. (A counterexample to \cref{thm: circulant prime order} exists for non-prime $q$.
%, the proof only holds for circulants of prime order. 
%We begin with basic propositions of circulant matrices
Omitted proofs are in \autoref{appendix: Circulant}):

\begin{definition}
    A $q\times q$ (real) matrix $C$ is called circulant, if each row in $C$ is obtained from the row above by a 
    cyclic shift of one position.
    We denote by $\mathcal{C}(c_0,c_1,\ldots,c_{q-1})$, the $q\times q$ circulant matrix with the first row $(c_0,c_1,\ldots,c_{q-1})$.
\end{definition}

\begin{proposition}[\cite{Circulant_Matrices}\label{lm: evs of circulant matrix}]
Let $C=\mathcal{C}(c_0,c_1,\ldots,c_{q-1})$.
Then $C$ has eigenvalues $(\psi_0,\psi_1,\ldots,\psi_{q-1})$, where
    $
    \psi_m=\sum_{j=0}^{q-1}c_j\zeta_q^{mj},\,\forall m\in [q].
    $
    Here $\zeta_q=e^{-{2\pi \mathfrak{i}}/{q}}$ is a $q$-th root of unity.
\end{proposition}

\vspace{-.1in}

\begin{restatable}{proposition}{PropertiesOfCirculant}
\label{lm: properties of circulant matrix}
    Let $C=\mathcal{C}(c_0,c_1,\ldots,c_{q-1})\in \Sym{q}{pd}{\Rp}$.
    Then $c_0>c_j$ for every $1\le j\le q-1$, and $\psi_0=\sum_{j=0}^{q-1}c_j$ is the unique largest eigenvalue of $C.$
\end{restatable}

\begin{lemma}\label{lm: lattice of c contained in Potts after thickening}
    Suppose $q\ge 3$ is prime.
    Let $C=\mathcal{C}(c_0,c_1,\ldots,c_{q-1})\in \Sym{q}{pd}{\Rp}$, and $\mathcal{T}_C(\theta)=\mathcal{C}(c_0^\theta,c_1^\theta,\ldots,c_{q-1}^\theta)$.
    There exists some $\theta\in \Rp$ such that $\mathcal{T}_C(\theta)\in \Sym{q}{pd}{\Rp}$ and
    $$
    \mathfrak{L}(\psi_0(\theta),\psi_1(\theta),
    \ldots,\psi_{q-1}(\theta))\subseteq\LPotts{q},
    $$
    where $\psi_m(\theta)=\sum_{j=0}^{q-1}c_j^\theta\zeta_q^{mj}.$
\end{lemma}
\begin{proof}   
    It is clear that $\mathcal{T}_C(\theta)\in \Sym{q}{}{\Rp}$ is still circulant for any $\theta\in \Rp.$
    By the main theorem in \cite{FITZGERALD1977633}, $\mathcal{T}_C(\theta)$ is positive definite for every $\theta\ge q-2.$
    For $\mathbf{x}=(x_0,x_1,\ldots,x_{q-1})\in \Z^q$ such that $\sum_{m=0}^{q-1}x_m=0$, define
    $
    f_{\mathbf{x}}(\theta)=\sum_{m=0}^{q-1}x_m\ln \psi_{m}(\theta)$ for $\theta\ge q-2.
    $
    The condition $\mathbf{x}\in\mathfrak{L}(\psi_0(\theta),\psi_1(\theta),\ldots,\psi_{q-1}(\theta))$ is equivalent to $f_{\mathbf{x}}(\theta)=0$.
    We next show that $f_{\mathbf{x}}(\theta)$ is not identically $0$ unless $x_0=0$.
    Let $c=\max_{1\le j\le q-1}c_j$, and $S=\{j\mid 1\le j\le q-1,c_j=c\}.$
    Define $\epsilon={c}/{c_0}$ and $\epsilon_j={c_j}/{c_0}$ for $1\le j\le q-1$.
    By \cref{lm: properties of circulant matrix}, $\epsilon,\epsilon_j\in(0,1).$
    Denote $\zeta_q$ by $\zeta.$
    A direct calculation shows that
    \begin{equation*}
        \begin{aligned}
        f_{\mathbf{x}}(\theta)
        =&\sum_{m=0}^{q-1}x_m\ln \left(\sum_{j=0}^{q-1}c_j^\theta \zeta^{mj}\right)
        =\sum_{m=0}^{q-1}x_m\left(\theta\ln c_0+\ln \left(1+\sum_{j=1}^{q-1}\epsilon_j^\theta\zeta^{mj}\right)\right)\\
        =&\sum_{m=0}^{q-1}x_m\ln \left(1+\sum_{j\in S}\zeta^{mj}
        \epsilon^\theta+o(\epsilon^\theta)\right)
        =\sum_{m=0}^{q-1}x_m \left(\sum_{j\in S}\zeta^{mj}
        \epsilon^\theta+o(\epsilon^\theta)\right)\\
        =&\sum_{m=0}^{q-1}x_m \sum_{j\in S}\zeta^{mj}
        \epsilon^\theta+o(\epsilon^\theta),\text{ as }\theta\to + \infty.\\
        \end{aligned}
    \end{equation*}
    \iffalse
    \begin{equation}
        \begin{aligned}
        f_{\mathbf{x}}(\theta)
        =&\sum_{j=0}^{q-1}x_j\left(\ln \left(\sum_{m=0}^{q-1}\psi_m^\theta \zeta^{jm}\right)-\ln q\right)\\
        =&\sum_{j=0}^{q-1}x_j\left(\theta\ln\psi_0+\ln \left(1+\sum_{m=1}^{q-1}\frac{\psi_m^\theta}{\psi_0^\theta}\zeta^{jm}\right)\right)\\
        =&\sum_{j=0}^{q-1}x_j\ln \left(1+\sum_{m\in S}\zeta^{jm}
        \left(\frac{\psi}{\psi_0}\right)^\theta+o\left(\frac{\psi}{\psi_0}\right)^\theta\right)\\
        =&\sum_{j=0}^{q-1}x_j \left(\sum_{m\in S}\zeta^{jm}
        \left(\frac{\psi}{\psi_0}\right)^\theta+o\left(\frac{\psi}{\psi_0}\right)^\theta\right)\\
        =&\sum_{j=0}^{q-1}x_j \sum_{m\in S}\zeta^{jm}
        \left(\frac{\psi}{\psi_0}\right)^\theta+o\left(\frac{\psi}{\psi_0}\right)^\theta,\text{ as }\theta\to + \infty\\
        \end{aligned}
    \end{equation}
    \fi
    
    We denote by $a_\theta$ the coefficient of $\epsilon^\theta$ in $f_{\mathbf{x}}(\theta)$.
    Re-arrange $a_\theta$ as $a_\theta=\sum_{k=0}^{q-1}y_k\zeta^k$,
    where $y_k=\sum_{m=0}^{q-1}x_m\sum_{j\in S}\mathbb{I}[mj\equiv k \text{ (mod }q)]\in \Z$.
    Note that $y_0=|S|\cdot x_0$.
    Suppose $f_{\mathbf{x}}(\theta)\equiv0$.
    Then $a_\theta=0$.
    %Since we assume that $q$ is prime, 
    For $q$ prime, 
    $x^{q-1} + \ldots + x +1 \in \Q[x]$ is irreducible, and $\zeta$ is a root. 
    $[\Q(\zeta):\Q]=q-1$ and
    $\{\zeta^k\}_{k=0}^{q-2}$ is a basis for $\Q(\zeta)$ over $\Q$.
    Then $a_\theta=\sum_{k=0}^{q-1}y_k\zeta^k=0\Rightarrow \sum_{k=0}^{q-2}(y_k-y_{q-1})\zeta^k=0 \Rightarrow y_0=y_1=\ldots=y_{q-1}$.
    However,
    $
    \sum_{k=0}^{q-1}y_k=\sum_{m=0}^{q-1}x_m\sum_{j\in S}1=|S|\cdot\sum_{m=0}^{q-1}x_m=0.
    $
    So all $y_i =0$.
    %$y_0=y_1=\ldots=y_{q-1}=0.$
    It follows that $x_0=0$.
    Thus, if $(\sum_{m=0}^{q-1}x_m=0)\land (x_0\neq 0)$, then $f_{\mathbf{x}}(\theta)$ is not identically 0.
    Since $f_{\mathbf{x}}(\theta)$ is an analytic function of $\theta$,
    the set of zeros of $f_{\mathbf{x}}(\theta)$ is countable. 
    %\textcolor{red}{Need a citation here.}
    %%% ZZT: I think it is well-known.
    So the union  of sets of zeros of $f_{\mathbf{x}}(\theta)$ over all $\mathbf{x}\in \Z^q$ such that $(\sum_{m=0}^{q-1}x_m=0)\land(x_0\neq 0)$ is also countable.
    Hence there exists some $\theta\ge q-2$ such that $f_{\mathbf{x}}(\theta)\neq 0$ for all $\mathbf{x}\in \mathbb{Z}^q$ with $\sum_{m=0}^{q-1}x_m=0$ and $x_0\neq 0$.
    For this $\theta$, $\mathbf{x}\in\mathfrak{L}(\psi_0(\theta),\psi_1(\theta),
    \ldots,\psi_{q-1}(\theta))\implies (\sum_{m=0}^{q-1}x_m=0)\land (f_\mathbf{x}(\theta)=0)\implies (x_0= 0)\land(\sum_{m=1}^{q-1}x_m=0)\implies \mathbf{x}\in \LPotts{q}.$
    Therefore,
    $
    \mathfrak{L}(\psi_0(\theta),\psi_1(\theta),
    \ldots,\psi_{q-1}(\theta))\subseteq \LPotts{q}.
    $
\end{proof}

%%% JYC: maybe add a one line remark of an explicit example for q not a prime that
% a counter example exists.

\vspace{-.1in}

\begin{theorem}\label{thm: circulant prime order}
    Let $q\ge 3$ prime, $C=\mathcal{C}(c_0,c_1,\ldots,c_{q-1})\in \Sym{q}{pd}{\Rp}$.
    Then $\PlGH(C)$ is \numPhard.
\end{theorem}
\begin{proof}
    Pick a $\theta$ in \cref{lm: lattice of c contained in Potts after thickening}.
    Then $\mathfrak{L}(\psi_0(\theta),\psi_1(\theta),
    \ldots,\psi_{q-1}(\theta))\subseteq\LPotts{q}$, where $\psi_m(\theta)=\sum_{j=0}^{q-1}c_j^\theta\zeta_q^{mj}\,(m\in[q])$ are eigenvalues of $\mathcal{T}_C(\theta)$.
    By \cref{lm: properties of circulant matrix}, $\psi_0(\theta)$ is the unique largest eigenvalue of $\mathcal{T}_C(\theta)$.
    So by \cref{thm: only look at lattice of eigenvalues}, $\PlGH(\mathcal{T}_C(\theta))$ is \#P-hard.
    By \cref{lemma: simpleThickening}, $\PlGH(\mathcal{T}_C(\theta))\le_p^T \PlGH(C)$.
    It follows that $\PlGH(C)$ is \#P-hard.
\end{proof}

%\begin{remark}
%    %We note that f
%    For non-prime $q$, a counterexample to \cref{thm: circulant prime order} 
%    exists $\mathcal{C}(4,2,1,2)  \in \Sym{4}{pd}{\Rp}$.
   % For example, consider $M = \left[\begin{smallmatrix}
   % 2 & 1\\
   % 1 & 2\\
   % \end{smallmatrix}\right]\otimes\left[\begin{smallmatrix}
   % 2 & 1\\
   % 1 & 2\\
   % \end{smallmatrix}\right]$.
   % Then after swapping the 3rd and 4th rows and columns, $M=\mathcal{C}(4,2,1,2)\in \Sym{4}{pd}{\Rp}.$
   % Nevertheless, $\PlGH(M)$ is tractable by the FKT algorithm together with a holographic transformation.
%\end{remark}
\section{Tensor Product of 2 by 2 Matrices}\label{sec: tensor of 2 by 2}
In this section, we assume that $M\in \Sym{q}{pd}{\Rp}$
%, and $M$ 
is a tensor product of $2\times 2$ matrices. We prove:
%We will prove the following theorem:

\begin{restatable}
{theorem}{HardnessForTensorTwoByTwo}\label{thm: hardness for 2 by 2 tensor positive definite}
    Let $M\in \Sym{q}{pd}{\Rp}$.
    If
    $
    M=\bigotimes_{i=1}^s A_i,
    $
    where $A_i\in \R^{2\times 2}$ for every $1\le i\le s$, 
    then $\PlGH(M)$ is $\#$P-hard unless $(A_i)_{11}=(A_i)_{22}$ for every $1\le i\le s$, in which case each $A_i$ is a matchgate and  the problem is P-time tractable by a holographic reduction to the FKT algorithm.
\end{restatable}

The proof of \cref{thm: hardness for 2 by 2 tensor positive definite} proceeds by successively refining the properties of $(A_i)_{1\le i\le s}$, ensuring that each step maintains the existing hardness.
%of $\PlGH(M)$.
First we have the following refinement:

\begin{restatable}{lemma}{TensorProductFirstReduction}\label{lm: tensor product of 2 by 2 first reduction}
    To prove \cref{thm: hardness for 2 by 2 tensor positive definite}, if suffices to prove that $\PlGH(M)$ is \#P-hard for $M=\bigotimes_{i=1}^sA_i$, where $A_i=\left[\begin{smallmatrix}
    1 & b_i\\
    b_i & c_i
\end{smallmatrix}\right]\in \Sym{2}{pd}{\Rp}$ such that $1>b_i>c_i>b_i^2>0$ for every $1\le i\le s$.
%%% JYC: I added:
%Being positive definite, $b_i^2 < c_i$.
\end{restatable}

In the following, we assume that $A_i$ $(\forall 1\le i\le s)$ satisfies the conditions in \cref{lm: tensor product of 2 by 2 first reduction}.

\begin{definition}\label{def: 2 by 2 matrix parameters}
    For $A=\left[\begin{smallmatrix}
    a & b\\
    b & c
\end{smallmatrix}\right]\in \Sym{2}{pd}{\Rp}$,
denote by $\lambda(A)$ and $\mu(A)$ its two eigenvalues, where

\begin{equation*}
    \begin{aligned}
        &\lambda(A):=\frac{1}{2}(a+c+\sqrt{(a-c)^2+4b^2}),
        \quad \mu(A):=\frac{1}{2}(a+c-\sqrt{(a-c)^2+4b^2}).\\
    \end{aligned}
    \vspace{-.1in}
\end{equation*}

Define
$\rho(A):=\frac{\mu(A)}{\lambda(A)}
%=\frac{a+c-\sqrt{(a-c)^2+4b^{2}}}{a+c+\sqrt{(a-c)^2+4b^{2}}}
$
and
$\gamma(A):=-\frac{a-\lambda(A)}{a-\mu(A)}.$ (They are invariant under a scalar multiple of $A$.)
\end{definition}

\begin{restatable}{lemma}{RhoGammaInZeroOne}\label{lm: rho and gamma range}
    For $A=\left[\begin{smallmatrix}
    a & b\\
    b & c
\end{smallmatrix}\right]\in \Sym{2}{pd}{\Rp}$ with $a>c$, we have $\rho(A),\gamma(A)\in(0,1).$
\end{restatable}

%\vspace{-.1in}

 \begin{definition}
        Denote the distinct elements of $\{A_i\}_{1\le i\le s}$ by $B_1, \dots, B_m$.
        Suppose $B_i=A_{p_i}$ for $1\le i\le m$, where $1\le p_1<p_2<\ldots<p_m\le s$.
    \end{definition}

\vspace{-.1in}
\begin{restatable} 
{lemma}{RhoDistinctByThickening}\label{lm: rho distinct by thickening}
    There exists some $n\in \Zp$, such that $\rho(\mathcal{T}_{B_i}(n))$  are all distinct for $1\le i \le m$. 
\end{restatable}
\begin{proofsketch}
    The idea is to look at the Taylor expansion of $\rho(\mathcal{T}_{B_i}(n))$ as $n\to + \infty.$
    Suppose $B=\left[\begin{smallmatrix}
    1 & b\\
    b & c
    \end{smallmatrix}\right]\in \Sym{2}{pd}{\Rp}$ such that $1>b>c>b^2$.
    By asymptotic analysis, we can show
    \begin{equation}\label{eq: rho(n)}
        \rho(\mathcal{T}_{B}(n))=c^n-b^{2n}+\Theta(c^nb^{2n}), \text{ as $n\to+\infty$.}
    \end{equation}
    
    Since the pairs $(b_{p_i},c_{p_i})$ are all distinct for every $1\le i<j\le m$, $\rho(\mathcal{T}_{B_i}(n))\neq\rho(\mathcal{T}_{B_j}(n))$ for all sufficiently large $n$ and every $1\le i<j\le m$.
    The lemma follows.
\end{proofsketch}

Fix the $n$ in \cref{lm: rho distinct by thickening}.
If $1>b_i>c_i>b_i^2>0$, then after thickening, we still have $1>b_i^n>c_i^n>b_i^{2n}>0$.
So thickening doesn't break our previous assumptions to $A_i$.
Also $\PlGH(\bigotimes_{i=1}^s \mathcal{T}_{A_i}(n))\equiv_p^T \PlGH(\mathcal{T}_M(n))\le_p^T \PlGH(M)$ by \cref{lemma: simpleThickening}.
Thus, after \cref{lm: rho distinct by thickening}, we may further assume that $\rho(B_i)$ are all distinct for $1\le i \le m$.
%Now we are ready to prove the following lemma:
\vspace{-.1in}

\begin{definition}
    For every $1\le i\le s$ and $\theta\in\R,$ define 
    $
    b_i(\theta)=\frac{\sqrt{\gamma_i}(1-\rho_i^\theta)}{1+\gamma_i\rho_i^\theta},
    $
    $
    c_i(\theta)=\frac{\rho_i^\theta+\gamma_i}{1+\gamma_i\rho_i^\theta}$,
    where $\rho_i=\rho(A_i),\gamma_i=\gamma(A_i)$.
    Define
    $A_i(\theta)=\left[\begin{smallmatrix}
    1 & b_i(\theta)\\
    b_i(\theta) & c_i(\theta)
\end{smallmatrix}\right]$.
\end{definition} 

It can be shown that for every $1\le i\le s$, $\mathcal{S}_{A_i}(\theta)=A_i(\theta)$ up to a scalar (\cref{lm: b(theta) and c(theta) after stretching}, \autoref{Appendix: Tensor of 2by2}).

\begin{lemma}\label{lm: c has no relation after stretching}
    There exists some $\theta\in \mathbb{R}$ such that $\widehat{\mathfrak{L}}(c_{p_1}(\theta),c_{p_2}(\theta),\ldots,c_{p_s}(\theta))=\{\mathbf{0}\}$.
\end{lemma}
\begin{proof}
Notice $\lim_{\theta\to + \infty}c_i(\theta)=\gamma_i\in (0,1)$. So there exists $U>0$, such that $c_i(\theta)>0$ for every $1\le i\le s$ and $\theta>U$.
For any $\mathbf{x}=(x_1,\ldots,x_s)\in \mathbb{Z}^m\setminus\{\mathbf{0}\}$, define $f_\mathbf{x}(\theta):=\sum_{i=1}^m x_i\ln c_{p_i}(\theta),\theta>U.$
The condition 
%that 
$\prod_{i=1}^m c_{p_i}(\theta)^{x_i}= 1$ is equivalent to $f_\mathbf{x}(\theta)=0.$
Next we show that for any $\mathbf{x}\in \Z^m\setminus \{\mathbf{0}\}$, $f_{\mathbf{x}}(\theta)\not\equiv 0.$  Suppose, for a contradiction, $f_{\mathbf{x}}(\theta)\equiv 0.$
We 
%first calculate the Taylor expansion of 
expand $\ln c_{p_i}(\theta)$ ``at $+\infty$''.
%$\ln c_{p_i}(k)$.
For simplicity, we temporarily omit subscripts.
We have $\ln c(\theta)=\ln(\rho^\theta+\gamma)-\ln({1+\gamma\rho^\theta}).$
Consider the real function $g(x)=\ln(x+\gamma)-\ln(\gamma x+1)$ 
defined for $x$ in a small interval containing $0$.
Note that $\rho_i\in(0,1)$ for every $1\le i\le m$ by \cref{lm: rho and gamma range}.
By Taylor expansion %$g(x)=g(0)+g'(0)x+\Theta(x^{2})$ as $x\to 0$ 
we have
%Now restrict $x$ to a small interval $(0, \delta)$ we have

\begin{equation}\label{eq: taylor expansion of c_k}
    \ln c(\theta)=g(\rho^\theta)=\ln\gamma+(\gamma^{-1}-\gamma)\rho^{\theta}+\Theta(\rho^{2\theta}), \text{ as } \theta\to +\infty.
\end{equation}

Plug \cref{eq: taylor expansion of c_k} into $f_\mathbf{x}(\theta)\equiv 0$, we obtain
    \begin{equation}\label{eq: taylor expansion of lattice condition of c}
    \sum_{i=1}^m x_i\ln\gamma_{p_i}+\sum_{i=1}^m x_i(\gamma_{p_i}^{-1}-\gamma_{p_i})\rho_{p_i}^{\theta}+\sum_{i=1}^m x_i\Theta(\rho_{p_i}^{2\theta}) \equiv 0,\text{ as } \theta\to +\infty.
    \end{equation}

    If $\sum_{i=1}^sx_i\ln\gamma_{p_i}\neq 0$,
    then the first term in \cref{eq: taylor expansion of lattice condition of c} is a non-zero constant, and the remaining terms decay exponentially as $\theta\to +\infty$.
    So \cref{eq: taylor expansion of lattice condition of c} fails for all sufficiently large $\theta$. 
    %Thus, $f_\mathbf{x}(\theta)\not \equiv 0.$
    A contradiction.
    If $\sum_{i=1}^sx_i\ln\gamma_i=0$, then
    $f_\mathbf{x}(\theta) \equiv 0$ is equivalent to
    {\small
    \begin{equation}\label{eq: simplified taylor expansion of lattice condition}
        \sum_{i=1}^sx_i(\gamma_{p_i}^{-1}-\gamma_{p_i})\rho_{p_i}^{\theta}+\sum_{i=1}^s x_i\Theta(\rho_{p_i}^{2\theta})
        \equiv 0,\text{ as } \theta\to +\infty.
    \end{equation}
    }
%    Given that the
As $\rho_{p_i}$ are distinct and $\mathbf{x} \neq \mathbf{0}$, there exists a unique index $i_0$ corresponding to the largest $\rho_{p_i}$ for which $x_i \neq 0$. 
    Then \cref{eq: simplified taylor expansion of lattice condition} becomes
    $
    x_{i_0}(\gamma_{p_{i_0}}^{-1}-\gamma_{p_{i_0}})\rho_{p_{i_0}}^\theta+o(\rho_{p_{i_0}}^\theta) \equiv 0,\text{ as } \theta\to +\infty,
    $
    a contradiction.
    %This is also a contradiction.  
    
    So far, we proved that for any $\mathbf{x}\in \Z^m\setminus\{\mathbf{0}\}$, $f_\mathbf{x}(\theta)\not\equiv 0.$
    Then the number of zeros of $f_\mathbf{x}(\theta)$ is countable for any $\mathbf{x}\in \Z^m\setminus\{\mathbf{0}\}$.
    Furthermore, the union of the sets of zeros over all $\mathbf{x}\in \Z^m\setminus\{\mathbf{0}\}$ is also countable. 
    Thus, there exists $\theta \in \mathbb{R}$ such that $f_\mathbf{x}(\theta) \neq 0$ for all $\mathbf{x} \in \mathbb{Z}^s \setminus \{\mathbf{0}\}$. 
    That is,
    $\widehat{\mathfrak{L}}(c_{p_1}(\theta), c_{p_2}(\theta), \ldots, c_{p_s}(\theta)) = \{\mathbf{0}\}.$
\end{proof}

Notice that in the proof of \cref{lm: c has no relation after stretching}, we can choose $\theta$ arbitrarily large.
Specifically, we can choose $\theta>0$ such that $1>b_i(\theta)>c_i(\theta)$ for every $1\le i\le s$ since $\lim_{\theta\to +\infty}b_i(\theta)=\sqrt{\gamma_i}>\gamma_i=\lim_{\theta\to +\infty}c_i(\theta)$.
Also, stretching has no affect on the positive definiteness, so the assumption $c_i(\theta)>b_i(\theta)^2$ is preserved.
Since $\PlGH(\bigotimes_{i=1}^sA_i(\theta))\equiv_p^T\PlGH(\mathcal{S}_M(\theta))\le_p^T \PlGH(M)$ by \cref{lemma: stretchingLemma}, we may replace every $A_i$ by $A_i(\theta)$ and further assume that $\widehat{\mathfrak{L}}(c_{p_1},c_{p_2},\ldots,c_{p_s})=\{\mathbf{0}\}$.

\begin{lemma}\label{lm: rho has no relation after thickening}
    There exists some $\eta\in \mathbb{R}$ such that $\widehat{\mathfrak{L}}(\rho(\mathcal{T}_{B_{p_1}}(\eta)),\rho(\mathcal{T}_{B_{p_2}}(\eta)),\ldots,\rho(\mathcal{T}_{B_{p_m}}(\eta)))=\{\mathbf{0}\}$.
\end{lemma}
\begin{proof}
    Fix any $\mathbf{x} \in \mathbb{Z}^m \setminus \{\mathbf{0}\},$ define $f_\mathbf{x}(\eta):=\sum_{i=1}^m x_i\ln \rho(\mathcal{T}_{B_{p_i}}(\eta)),\eta>0.$
    The condition $\mathbf{x}\in\widehat{\mathfrak{L}}(\rho(\mathcal{T}_{B_{p_1}}(\eta)),\rho(\mathcal{T}_{B_{p_2}}(\eta)),\ldots,\rho(\mathcal{T}_{B_{p_m}}(\eta)))$ is equivalent to $f_\mathbf{x}(\eta)=0$.
    We show that $f_\mathbf{x}(\eta)\not\equiv 0$.
    Recall \cref{eq: rho(n)}:
    $\rho(\mathcal{T}_{B}(\eta)) = c^\eta - b^{2\eta} + \Theta(c^\eta b^{2\eta})$ as $\eta\to + \infty$, for $B=\left[\begin{smallmatrix}
    1 & b\\
    b & c
    \end{smallmatrix}\right]$ such that $1>b>c>b^2$.
    Using the approximation $\ln(1-u) = -u + o(u)$ as $u \to 0$, we obtain

    \begin{equation}\label{eq: log mu expansion}
        \ln\rho(\mathcal{T}_{B}(\eta)) = \ln\left( c^\eta \left( 1 - \frac{b^{2\eta}}{c^\eta} + \Theta(b^{2\eta}) \right) \right)
        = \eta\ln c - \frac{b^{2\eta}}{c^\eta} + o\left(\frac{b^{2\eta}}{c^\eta}\right).
    \end{equation}
    
    Substituting \cref{eq: log mu expansion} into $f_\mathbf{x}(\eta)=0$ yields:
    {\small
    \begin{equation}\label{eq: asymptotic balance}
        \eta \sum_{i=1}^m x_i \ln c_{p_i} = \sum_{i=1}^m x_i\cdot \frac{b_{p_i}^{2\eta}}{c_{p_i}^{\eta}}
        + \sum_{i=1}^m x_i\cdot o\left(\frac{b_{p_i}^{2\eta}}{c_{p_i}^\eta}\right).
    \end{equation}
    }
    Given the assumption $\widehat{\mathfrak{L}}(c_{p_1}, \ldots, c_{p_s}) = \{\mathbf{0}\}$, we know $\sum_{i=1}^m x_i \ln c_{p_i}\neq 0$. 
    The LHS of \cref{eq: asymptotic balance} grows linearly with $\eta$, whereas the RHS decays exponentially.
    Thus, \cref{eq: asymptotic balance} fails for all sufficiently large $\eta$.
   So, indeed $f_\mathbf{x}(\eta)\not\equiv 0.$
    By the same counting argument as in \cref{lm: c has no relation after stretching}, there exists $\eta \in \mathbb{R}$ such that $f_\mathbf{x}(\eta) \neq 0$ for all $\mathbf{x} \in \mathbb{Z}^m \setminus \{\mathbf{0}\}$. 
    That is, $\widehat{\mathfrak{L}}(\rho(\mathcal{T}_{B_{p_1}}(\eta)),\rho(\mathcal{T}_{B_{p_2}}(\eta)),\ldots,\rho(\mathcal{T}_{B_{p_m}}(\eta)))=\{\mathbf{0}\}$.
\end{proof}

An implication of \cref{lemma: latticeInterpolation} is that if the spectral lattice of $M$ is trivial, then $\PlGH(M)$ is \numPhard.
(Basically, we can interpolate 
%anything 
freely by \cref{lemma: latticeInterpolation}.)
When $M$ is a tensor product of $2 \times 2$ matrices, we can show that the hardness criterion extends to the triviality of the ``reduced spectral lattice".
Formally, we have the following lemma which is proved in \autoref{Appendix: Tensor of 2by2}:
\begin{restatable}
{lemma}{ReducedLatticeLemma}\label{lm: reduced lattice lemma}
    If $\widehat{\mathfrak{L}}(\rho(B_1),\rho(B_2),\ldots,\rho(B_m))=\{\mathbf{0}\},$ then $\PlGH(M)$ is $\#$P-hard.
\end{restatable}

Finally, we are ready to prove \cref{thm: hardness for 2 by 2 tensor positive definite}:

%\HardnessForTensorTwoByTwo*
\begin{proof}
    Fix the $\eta$ in \cref{lm: rho has no relation after thickening}, $\PlGH(\mathcal{T}_M(\eta))\le_p^T \PlGH(M)$ by \cref{lemma: simpleThickening}.
    %Then apply \cref{lm: reduced lattice lemma}.
    %
    Then the theorem follows from \cref{lm: reduced lattice lemma}.
\end{proof}

We can also extend  the dichotomy from $\Sym{q}{pd}{\Rp}$ to
$\Sym{q}{}{\R_{\geq 0}}$; see \autoref{Appendix: Tensor of 2by2}.
%with a little effort.
%The proof is deferred to \autoref{Appendix: Tensor of 2by2}.
\begin{restatable}
{theorem}{HardnessForTensorTwoByTwoNonnegative}\label{thm: hardness for 2 by 2 tensor nonnegative}
    Let $M\in \Sym{q}{}{\R_{\ge 0}}$ be a non-zero matrix.
    Suppose
    $M = \bigotimes_{i=1}^{s} A_{i}$,
    where $A_i\in \Sym{2}{}{\R_{\ge 0}}$ for every $1\le i\le s$.
    Let $I := \{i \in[s] : (A_{i})_{12}\neq 0 
    \text{ and } \det(A_{i})\neq 0\}$.
    Then $\PlGH(M)$ is $\#$P-hard unless
    $(A_{i})_{11}=(A_i)_{22}$ for every $i\in I$,
    in which case the problem is computable in P-time.
    %tractable.
\end{restatable}

\newpage

\appendix
\noindent\textbf{AI Disclosure:}
All the results in this paper were obtained before February 2026, and the entire manuscript was written by human authors.
We used Gemini 3.1 Pro to assist with literature search on matrix perturbation theory \cite{kato1995perturbation}. 
Gemini 3.1 Pro also helped with a routine calculation within the proof of \cref{lm: b(theta) and c(theta) after stretching}, with no real intellectual ingredient, and is manually verified. 
The tool materially affected \autoref{Appendix: Tensor of 2by2}. The authors verified the correctness and originality of all content including references.
\section{A Very Brief Introduction to
Quantum Automorphism Groups and the
``Quantum Barrier''}
\label{Appendix: Quantum}

In this section,
we briefly describe \emph{quantum automorphism groups}~\cite{wang1995free, wang1998quantum,
bichon2003quantum, banica2005quantum,
atserias2019quantum, chassaniol2019study,
lupini2020nonlocal, manvcinska2020quantum, kar2026npa}
and their relation to the ``quantum barrier'' discovered in
\cite{cai2026dichotomy}, which motivates this work. We begin by introducing the necessary
operator-algebraic framework before giving the formal definitions.

A $C^{*}$-algebra is a Banach algebra together with an involution $^*$ satisfying the properties of the adjoint.
It is a generalization of certain classes of bounded linear operators on Hilbert spaces. They form the analytic foundation of compact
quantum groups~\cite{arveson1998invitation}. 
%You can learn more about them from
We consider a matrix
$\mathcal{U} = (u_{ij})$ whose entries
belong to a (not necessarily commutative)
unital $C^{*}$-algebra.
The \emph{conjugate transpose}
$\mathcal{U}^{\dagger}$ is defined by
$(\mathcal{U}^{\dagger})_{ij} = u_{ji}^{*}$.

\begin{definition}[Wang \cite{wang1998quantum}]
A matrix $\mathcal{U} = (u_{ij})$ with entries in a unital $C^{*}$-algebra
is called a \emph{quantum permutation matrix} if, for all $i,j$,
    \begin{enumerate}
        \item $u_{ij} = u_{ij}^{*} = u_{ij}^{2}$,
        \item $\sum_{j} u_{ij} = 1$ and $\sum_{i} u_{ij} = 1$.
    \end{enumerate}
\end{definition}
These are ``quantum'' generalizations of permutation matrices.
%We are now ready to define 
Quantum automorphism groups
%They 
arise as a special case of quantum permutation groups,
which were introduced by Wang~\cite{wang1998quantum} as a
noncommutative generalization of classical permutation groups.
Formally, quantum permutation groups are realized within the
operator-algebraic framework of compact matrix quantum groups
developed by Woronowicz~\cite{woronowicz1987compact}.

\begin{definition}[Wang \cite{wang1998quantum}]
    The \emph{quantum permutation group} $\mathcal{Q}$ of order $q$ is defined
    by the universal $C^{*}$-algebra $C(\mathcal{Q})$ generated by the
    entries of a $q \times q$ quantum permutation matrix $\mathcal U$.
    The matrix $\mathcal{U}$ is called the \emph{fundamental representation}
    of $\mathcal{Q}$.
\end{definition}

Using quantum permutation group one can define the notion of two graphs $H$ and $H'$ being quantum
isomorphic. 
When $H=H'$ this is called the quantum automorphism group.
\begin{definition}[Wang \cite{wang1998quantum}, Banica \cite{banica2005quantum}]
    Let $M \in \Sym{q}{}{\R}$.
    The \emph{quantum automorphism group} $\qut(M)$ is the quantum
    permutation group generated by a quantum permutation matrix
    $\mathcal U$ satisfying $\mathcal{U} M = M \mathcal{U}$.
\end{definition}

%%% JYC: add a word about def'n of quantum isomorphism between graphs 

Man{\v{c}}inska and Roberson~\cite{manvcinska2020quantum} established 
a direct link of quantum isomorphism between graphs and planar graph homomorphism.
They proved that two graphs $H$ and $H'$ are quantum isomorphic if and only if 
they define the same function $Z_{H}(G) = Z_{H'}(G)$
from all {\it planar} input graphs $G$. 
We do not delve further into the
formal theory of quantum automorphism
groups here.
Instead, we summarize the consequences of their appearance in
planar graph homomorphism counting, as established in
\cite{cai2026dichotomy}.
The following theorem \cite[Theorem~36]{cai2026dichotomy}
%(also proven in \cite{cai2024planar})
%%% JYC as this is just a brief recap, no need to be too complete in this. especially to cite ourselves
establishes the connection between $\qut(M)$ and
the set of {\it all} planar edge gadgets $\PlEdge(M)$ that can be constructed from $M$.

\begin{theorem}\label{theorem: qutM-PlEdge}
    Let $M \in \Sym{q}{}{\mathbb{R}_{\geq 0}}$. Then,
    $\mathbf{K}(M)_{ii} = \mathbf{K}(M)_{jj}$ for every
    $\mathbf{K}(M) \in \PlEdge(M)$ if and only if
    $i$ and $j$ are in the same orbit of $\qut(M)$.
\end{theorem}

In particular, if $\qut(M)$ is non-trivial, then there exist distinct
labels $i \neq j$ that lie in the same orbit of $\qut(M)$, and hence are
\emph{inseparable} by planar gadgets in the following concrete sense:
every planar edge gadget $\mathbf{K}(M)$ has equal diagonal entries at
those indices, $\mathbf{K}(M)_{ii} = \mathbf{K}(M)_{jj}$.
Therefore, any reduction strategy that relies on \emph{vertex separation}
(i.e., producing planar gadgets whose signatures have distinct diagonal
entries so as to isolate a smaller effective domain) cannot succeed in
general once $\qut(M)$ has nontrivial orbits; this is the ``quantum
barrier'' of \cite{cai2026dichotomy}.
%%% JYC: it will be good to make the claim that all previous proofs of
% planar GH dichotomy uses this "isolate a smaller effective domain" reduction.

This barrier can apply even when the classical automorphism group
$\aut(M)$ is trivial:
$\qut(M)$ can be strictly larger than $\aut(M)$, collapsing labels that
are classically distinguishable but remain indistinguishable to all
planar gadgets.
Moreover, \cite{cai2026dichotomy} shows
(via \cite{atserias2019quantum,slofstra2019set})
that determining whether $\qut(M)$ is trivial is \emph{undecidable}, ruling out
any general algorithmic test for whether planar vertex separation is
possible.

Crucially, the quantum barrier only rules out reductions
from smaller domain problems.
It does not imply that \#P-hardness is unattainable when
$\qut(M)$ is nontrivial.
In this paper, we specifically demonstrate this fact by proving \#P-hardness for 
circulant matrices.
Circulant
matrices have all identical diagonal entries,
and their quantum automorphism group is such that
\emph{all} elements lie on the \emph{same orbit}.
Consequently, given any circulant matrix $C$,
$\mathbf{K}(C)_{ii} = \mathbf{K}(C)_{jj}$ for all 
planar edge gadgets $\mathbf{K}$.
So, vertex separation techniques provably
cannot be used to find reductions from any
smaller-sized \#P-hard matrices.
Nevertheless, our circulant hardness result
(see \cref{thm: circulant prime order}) shows that planar edge
gadgets can still be used profitably when combined with additional
analytic and spectral arguments, 
to find reductions from same-sized
\#P-hard matrices.

While this \#P-hardness proof for circulant matrices
serves as a \emph{proof of concept} that
planar edge gadgets can establish \#P-hardness
beyond the quantum barrier,
it does not provide a \emph{general strategy}
for proving \#P-hardness for
matrices with non-trivial quantum automorphism groups.
The main challenge we face is that given any
matrix $N$ for which $\qut(N)$ is \emph{not} maximal,
there exist matrices $M$, such that
no planar edge gadget can be used to prove the reduction:
$\PlGH(N) \leq_{p}^{T} \PlGH(M)$.
In other words, if $\qut(N)$ is not maximal,
then $N$ \emph{cannot} serve as a universal source
of \#P-hardness for all matrices $M$.
Let us now consider the
maximal quantum automorphism group $S_{q}^{+}$.
Proposition~38 of \cite{cai2026dichotomy} identifies this regime
exactly:
$$\qut(M) = S_{q}^{+} \quad\Longleftrightarrow\quad
M \in \operatorname{span}_{\mathbb R}\{I, J\},$$
where $I$ is the identity matrix and $J$ is the all-1 matrix.
This is precisely the set of Potts-model matrices,
for which $\PlGH(M)$ already admits a dichotomy
(see \cite{vertigan2005computational}).
This sharp dichotomy motivates
a \emph{sufficient condition} for proving
the \#P-hardness of matrices
that operates \emph{within}
quantum-symmetry constraints rather than attempting to break them.
In this spirit, \cref{thm: only look at lattice of eigenvalues}
yields the following sufficient condition for proving \#P-hardness.

\begin{theorem}\label{thm: sufficientTest}
    Let $M \in \Sym{q}{pd}{\Rp}$. 
    %Then,
    Suppose  
    there exists some
    planar edge gadget $\mathbf{K}$ such that
    $\mathbf{K}(M)$ has eigenvalues
    $\lambda_{1} > \lambda_2 \ge\dots \ge \lambda_{q}$ that satisfy:
    $\mathfrak{L}(\lambda_1, \lambda_2, \dots,\lambda_q)
    \subseteq \LPotts{q}$. Then $\PlGH(M)$ is \#P-hard.
\end{theorem}

The advantage of \cref{thm: sufficientTest} is twofold.
First, it yields a simple, checkable criterion for establishing
\#P-hardness.
Second, unlike vertex-separation arguments,
it is \emph{not} restricted to matrices with
trivial quantum automorphism groups.
Instead, the condition operates directly
on the lattice of eigenvalues of $M$,
relating it to the eigenvalue lattice of the Potts model.
Consequently, the criterion remains applicable
even when $\qut(M)$ is non-trivial,
including the case of maximal quantum symmetry,
separation-based techniques provably fail by
\cref{theorem: qutM-PlEdge}.

It is important to note that while \cref{thm: sufficientTest}
is a \emph{sufficient} condition to prove \#P-hardness,
it is \emph{not} a necessary condition.
For a concrete example, assume $M = A_{1} \otimes A_{2}$, where
$A_{1} = \left[\begin{smallmatrix}
    a_{1} & b_{1}\\
    b_{1} & c_{1}\\
\end{smallmatrix}\right]$, and $ A_{2} = \left[\begin{smallmatrix}
    a_{2} & b_{2}\\
    b_{2} & c_{2}\\
\end{smallmatrix}\right]$.
If we let $\lambda_{1}, \lambda_{2}$ be the eigenvalues of $A_{1}$,
and $\mu_{1}, \mu_{2}$ be the eigenvalues of $A_{2}$,
it follows that $\lambda_{1}\mu_{1}, \lambda_{1}\mu_{2}, 
\lambda_{2}\mu_{1},  \lambda_{2}\mu_{2}$ are the
eigenvalues of $M$.
It is easily seen that $(1, -1, -1, 1) \in \mathfrak{L}
(\lambda_{1}\mu_{1}, \lambda_{1}\mu_{2}, 
\lambda_{2}\mu_{1},  \lambda_{2}\mu_{2})$.
So, the lattice of eigenvalues of $M$ is not a subset of
$\LPotts{q}$.
More worryingly, it can be seen that for any planar edge gadget
$\mathbf{K}$, $\mathbf{K}(M) = 
\mathbf{K}(A_{1}) \otimes \mathbf{K}(A_{2})$.
This ensures that for {\it any planar} edge gadget $\mathbf{K}$,
the lattice of eigenvalues of $\mathbf{K}(M)$ cannot be
a subset of $\LPotts{q}$.
In other words,  when $M$ is a tensor product,
\cref{thm: sufficientTest} cannot be used to
establish \#P-hardness.
On the other hand, we already have a known dichotomy
result for $4 \times 4$ matrices (see \cite{cai2024polynomial}),
that includes the case of tensor products,
as seen above.
So, it is in fact known that when $a_{1} \neq c_{1}$
or $a_{2} \neq c_{2}$, $\PlGH(M)$ is in fact
\#P-hard, even though $M$ does not meet
the requirements of \cref{thm: sufficientTest}.
This proof for hardness in that case relies a very delicate analysis coupled with
vertex separation techniques, thus raising a very natural question:
are there any matrices for which neither vertex separation
techniques, nor \cref{thm: sufficientTest}
can be used to prove \#P-hardnes?

To answer this question, we focus on matrices of the
form $M = A_1 \otimes \cdots \otimes A_s$ with
$2\times 2$ factors.
The multiplicative relations among eigenvalues are
so structured that the lattice condition
in \cref{thm: sufficientTest}
can fail to certify hardness
even when the problem is in fact \#P-hard.
At the same time, the tractable subfamily of problems
are exactly the ones where the
quantum barrier is most pronounced:
when every factor has equal diagonal
entries, the diagonal of $M$ is completely uniform
and nontrivial quantum
symmetries are present,
making vertex separation impossible.
In this setting, we prove a dichotomy,
which shows that, in this important family
where the spectral
sufficiency condition can be too rigid,
the remaining boundary is
nevertheless clean:
tractability coincides precisely with the
quantum symmetries, 
%where vertex separation techniques lead to a proof of \#P-hardness.
and our analytic approach and analysis on the ``reduced lattice" lead to a proof of \#P-hardness.
\section{Model of Computation}\label{appendix: modelComputation}

The Turing machine model is naturally suited to the study of computation over discrete structures such as integers or graphs. 
When  $M \in \Sym{q}{}{\mathbb{R}}$, for  $\PlGH(M)$ 
one usually restricts $M$ to be a matrix
with only algebraic numbers. This is strictly for the consideration
of the model of computation, even though allowing all
real-valued matrices would be more natural. 
 
 There is a formal (albeit nonconstructive) method to treat  $\PlGH(M)$ 
 for arbitrary real-valued matrices $M$ and yet stay strictly
 within the Turing machine model in terms of  bit-complexity.
In this paper, because our proof depends heavily on analytic
argument with continuous functions on 
%$\mathbb{R}^{d}$, 
$\mathbb{R}$, 
this
logical formal view becomes necessary.

To begin with, we recall a theorem from field theory:
Every extension field ${\bf F}$ over  ${\mathbb Q}$
by a finite set of real numbers is  a finite algebraic extension ${\bf E}'$
of a certain  purely transcendental   extension field ${\bf E}$
over ${\mathbb Q}$,
which has the form ${\bf E} = {\mathbb Q}(X_1,
\ldots, X_m)$ where $m \ge 0$ and $X_1,
\ldots, X_m$ are algebraically independent~\cite{jacobson1985basic} (Theorem 8.35, p.~512).
${\bf F}$ is said to have
a finite transcendence degree $m$ over ${\mathbb Q}$.
It is known that $m$ is uniquely defined for ${\bf F}$.
Since 
$\rm{char}~{\mathbb Q} =0$,
the finite algebraic extension ${\bf E}'$ over ${\bf E}$
is actually simple, ${\bf E}' = {\bf E}(\beta)$ for some $\beta$,
and it is specified by a
minimal polynomial in ${\bf E}[Y]$, the polynomial ring over ${\bf E}$ with one
indeterminant $Y$.
Now given a real matrix $M$, let ${\bf F} = {\mathbb Q}(M)$ 
be the extension field by adjoining the entries of $M$.
We  consider $M$ is fixed for the problem $\PlGH(M)$,
and thus we may assume (nonconstructively) that the form
${\bf F} = {\bf E}(\beta)$
 and ${\bf E} = {\mathbb Q}(X_1,
\ldots, X_m)$ are given. (This means, among other things,
that the minimal polynomial of $\beta$ over ${\bf E}$ is given,
and all arithmetic operations can be performed on ${\bf F}$.)

Now, the computational problem $\PlGH(M)$ is the following:
Given a planar 
$G$, compute $Z_{M}(G)$ as an element in ${\bf F}$
(which is expressed as a polynomial in $\beta$ with coefficients in ${\bf E}$).
More concretely, we can show that this is equivalent to the following problem $\COUNT(M)$:
The input  is a pair $(G,x)$,
  where $G=(V,E)$ is a planar graph and $x\in {\bf F}$.
The output is\vspace{-0.1cm}
$$
\text{\#}_{M}(G,x)= \Big|\big\{\sigma:V\rightarrow
  [q]\hspace{0.08cm}: \hspace{0.08cm} \prod_{(u, v) \in E} m_{\sigma(u), \sigma(v)}=x\big\}\Big|,\label{full_COUNTM}
  $$
a non-negative integer. Note that, in this definition,
we are basically combining terms with the same 
product value in the definition of $Z_{M}(G)$.

Let $n=|E|$.
Define  $X$ to be the  set of all possible product values
appearing in $Z_{M}(G)$:
\begin{equation}\label{full_definitionreuse}
X=\left\{\prod_{i,j\in [q]}m_{ij}^{k_{ij}}\hspace{0.08cm}\Big|\hspace{0.1cm}
  \text{integers $k_{ij}\ge 0$ and $\sum_{i,j\in [q]}k_{ij}=n$}
\right\}.
\end{equation}
There are $\binom{n+q^2-1}{q^2-1} = n^{O(1)}$ many integer sequences
$(k_{i,j})$ such that  $k_{i,j}\ge 0$ and $\sum_{i,j\in [q]}k_{i,j}=n$.
$X$
  is defined as a set, not a multi-set.
After removing repeated
elements the cardinality $|X|$ is also polynomial in $n$.
 For fixed and given ${\bf F}$
 the elements in $X$ can be enumerated in polynomial time in $n$.
 (It is important that ${\bf F}$ and $q$ are all treated as fixed constants.)
It then follows from the definition that
  $\text{\#}_{M}(G,x)= 0$ for any $x\notin X$.
This gives us the following relation:
\[Z_M(G)= \sum_{x\in X} x \cdot \text{\#}_{M}(G,x),\ \ \ \text{for any 
  graph $G$,}\]
and thus, $\PlGH(M)\le_p^T \COUNT(M).$

For the other direction,
we construct, for any $p\in [|X|]$ (recall that $|X|$ is polynomial in $n$), the graph $\mathbf{T_{p}}G$.
Then,
$$
Z_M(\mathbf{T_p}G) =
Z_{\mathbf{T_{p}}(M)}(G) = \sum_{x\in X} x^p \cdot \text{\#}_{M}(G,x),\ \ \ \text{for any 
  graph $G$.}
$$
This is a Vandermonde system; it has full rank since
elements in $X$ are distinct by definition. So by
querying $\PlGH(M)$ for the
  values of $Z_M(\mathbf{T_p}G)$,
  we can solve it in polynomial time
  and get $\text{\#}_{M}(G,x)$ for\vspace{0.0015cm} every non-zero $x\in X$.
To obtain $\text{\#}_{M}(G,0)$ (if $0\in X$), we note that
$$
\sum_{x\in X} \text{\#}_{M}(G,x) =q^{|V|}.
$$
This gives us a polynomial-time reduction and thus, $\COUNT(M)\le_p^T \PlGH(M)$.
We have proved
\begin{lemma}\label{lemma: full_count}
For any fixed  $M \in \Sym{q}{}{\mathbb{R}}$,  
   $\PlGH(M)\equiv_p^T \COUNT(M)$.
\end{lemma}
Thus, $\PlGH(M)$ can be identified with the 
problem of producing those polynomially many integer coefficients
in the canonical expression for $Z_M(G)$ as a sum
of (distinct) terms from $X$.

This  formalistic view has 
the advantage that we can treat the complexity 
of  $\PlGH(M)$ for  general $M$, and not restricted to algebraic numbers. 
 Thus, numbers such as $e$ or $\pi$
need not be excluded.
More importantly, in this paper this generality is essential, due to the proof techniques that we employ.
% Furthermore, once freed from this restriction we in fact explicitly use
% transcendental numbers as a tool in our proof (see \cref{lemma: analyticGadgets}).
In short, in this paper, treating the complexity of  $\PlGH(M)$ for  general real $M$
is not a \emph{bug} but a \emph{feature}.

However, we note that this treatment 
has the following subtlety.  
For the computational problem $\PlGH(M)$
the formalistic view demands that
${\bf F}$ be specified in the form ${\bf F} = {\bf E}(\beta)$.
Such a form exists, and its specification is of
constant size when  measured
in terms of the size  of the input graph $G$. 
However, 
in reality many basic questions for transcendental numbers
are unknown.
For example, it is still unknown whether $e + \pi$ or $e \pi$ are
rational, algebraic irrational or transcendental,
and it is open whether ${\mathbb Q}(e, \pi)$ has  transcendence degree
2 (or 1) over  ${\mathbb Q}$, i.e., whether $e$ and $\pi$ are algebraically
independent.
The formalistic view here non-constructively
assumes this information is given for ${\bf F}$.
A  polynomial time reduction $\Pi_1  \le_p^T \Pi_2$  from one problem
to another  in this setting merely
implies that the \emph{existence} of a polynomial time algorithm
for $\Pi_2$ logically implies the 
\emph{existence} of  a  polynomial time algorithm
for $\Pi_1$. We do not actually obtain such
an algorithm constructively. 

This logical detour not withstanding, if a reader
is  interested only in the complexity of $\PlGH(M)$ 
for $M \in \Sym{q}{}{\mathbb{Z}_{\geq 0}}$, then the
complexity dichotomy proved in this paper holds
according to the standard definition of $\PlGH(M)$ 
for integral $M$ in terms of
the model of computation; the fact that this is proved 
in a broader setting for all real matrices $M$ 
is irrelevant. This  is akin to the situation
in analytic number theory, where one might be 
interested in a question strictly about the ordinary
integers, but the theorems are proved
in a broader setting of analysis.
% \section{Planar Edge Gadgets}\label{appendix: edgeGadgetsAppendix}
\section{Planar Edge Gadgets}

In this section, we shall prove all omitted proofs
from \cref{sec: preliminaries}.
First, we prove \cref{lemma: MequivalentCM}.

\MequivalentCM*
\begin{proof}
    Note that given any $G = (V, E),$
    $$Z_{cM}(G) = \sum_{\sigma: V \rightarrow [q]}
    \prod_{(u, v) \in E}(cM)_{\sigma(u)\sigma(v)}
    = c^{|E|} \sum_{\sigma: V \rightarrow [q]}
    \prod_{(u, v) \in E}cM_{\sigma(u)\sigma(v)}
    = c^{|E|}Z_{M}(G).$$
    Therefore, oracle access to one of $\PlGH(M)$ or $\PlGH(cM)$
    gives the other.
   % allows us to solve $\PlGH(cM)$ (resp. $\PlGH(M)$).
   % Therefore, oracle access to $\PlGH(M)$ (or $\PlGH(cM)$)
   % allows us to solve $\PlGH(cM)$ (resp. $\PlGH(M)$).
    %%% JYC this does not need to be stated for both direction,
    %%% as one just take c^{-1}
\end{proof}

\subsection{Thickening Gadgets}\label{appendix: thickeningAppendix}

In this section, we shall explore the thickening gadgets,
and gadget interpolation more rigorously.
For $m, n \ge 1$, let
$$\mathcal{P}_{m}(n) = \left\{\mathbf{x} = (x_{i})_{i \in [m]}
\in (\mathbb{Z}_{\geq 0})^{m} \hspace{0.08cm}\Big|\hspace{0.1cm}
\sum_{i \in [m]}x_{i} = n\right\}.$$

We note that
given any planar graph $G = (V, E)$,
\begin{equation}\label{equation: thickeningEqn}
    Z_{M}(\mathbf{T_{n}}G) = Z_{\mathbf{T_{n}}(M)}(G) = 
    \sum_{x \in X(G)}x^{n} \cdot \text{\#}_{M}(G, x)
\end{equation}
where
\begin{equation}\label{equation: thickeningX(G)}
    X(G) = \left\{\prod_{i, j\in [q]}
    M_{ij}^{k_{ij}}\hspace{0.08cm}\Big|\hspace{0.1cm}
    \mathbf{k} = (k_{ij})_{i, j \in [q]} \in
    \mathcal{P}_{q^{2}}(|E|)\right\},
\end{equation}
and
\begin{equation}\label{equation: thickeningNumMappings}
    \text{\#}_{M}(G,x)= \Big|\big\{\sigma:V\rightarrow [q]\hspace{0.08cm}:
    \hspace{0.08cm} \prod_{(u, v) \in E} M_{\sigma(u)\sigma(v)}=x\big\}\Big|.
\end{equation}

Note that given any $x \in X(G)$,
$\text{\#}_{M}(G, x)$ does not depend on $n$,
but depends only on the entries of the matrix $M$.
Generating Sets, as defined in \cref{definition: generatingSet}
let us deal with this dependence.
First, we shall prove that generating sets always exist.

\generatingSet*

\begin{lemma}\label{lemma: generatingSet}
    Every finite set $\mathcal{A} \subset \mathbb{R}_{\neq 0}$
    of non-zero real numbers has
    a generating set.
\end{lemma}
\begin{proof}
    Consider the multiplicative group $\mathcal{G}$ 
    generated by 
    the positive real numbers 
    $\{|a| : a \in \mathcal{A}\}$.
    It is  a subgroup of the multiplicative group
    $(\mathbb{R}_{> 0}, \cdot)$.
    Since  $\mathcal{A}$ is finite, and 
    $(\mathbb{R}_{> 0}, \cdot)$ is torsion-free,
    the group $\mathcal{G}$ is  a finitely generated free Abelian group, 
    and thus isomorphic to 
    $\mathbb{Z}^d$ for some $d \ge 0$.
    Let $f$ be this isomorphism from $\mathbb{Z}^{d}$ to
    the multiplicative group.
    By flipping $\pm 1$ in $\mathbb{Z}$ we may assume that this isomorphism maps the basis
    elements of $\mathbb{Z}^{d}$ to some elements
    $\{g_{t}\}_{t \in [d]}$ such that $g_{t} > 1$
    for all $t \in [d]$.
    The set $\{g_{t}\}_{t \in [d]}$ is
    a generating set.
\end{proof}

We now use \cref{lemma: generatingSet} to find a generating set
for the entries $(M_{ij})_{i, j \in [q]}$ of any 
matrix $M \in \Sym{q}{}{\mathbb{R}_{\neq 0}}$.
Note that this generating set need not be unique. 
However, with respect to a fixed generating set, 
for any $M_{ij}$, there are unique integers 
$e_{ij0} \in \{0, 1\}$, and
$e_{ij1}, \dots, e_{ijd} \in \mathbb{Z}$,
such that
\begin{equation}\label{equation: generatingM}
    M_{ij} = (-1)^{e_{ij0}} \cdot g_{1}^{e_{ij1}} \cdots g_{d}^{e_{ijd}}.
\end{equation}

\begin{remark*}
    It should be noted that since $M$ is symmetric, $M_{ij} = M_{ji}$
    for all $i, j \in [q]$.
    The uniqueness of the integers $e_{ijt}$
    in \cref{equation: generatingM} then implies that
    for all $i, j \in [q]$,
    $e_{ijt} = e_{jit}$ for all $t \in [d]$.
\end{remark*}

We can now prove the following lemma.

\begin{lemma}\label{lemma: thickeningGeneralLemma}
    Let $M \in \Sym{q}{}{\mathbb{R}_{\neq 0}}$
    with its entries generated by some $\{g_t\}_{t \in [d]}$,
    such that
    $$M_{ij} = (-1)^{e_{ij0}} \cdot
    g_{1}^{e_{ij1}} \cdots g_{d}^{e_{ijd}}.$$
    Furthermore, assume that $e_{ijt} \geq 0$ for all
    $i, j \in [q]$, and $t \in [d]$.
    Then, for any
    $N \in \Sym{q}{}{\mathbb{R}}$
    of the form
    %$$N_{ij} = (-1)^{e_{ij0}} \cdot z_{1}^{e_{ij1}} \cdots z_{d}^{e_{ijd}},$$
    %%%%ZZT: To aviod 0^0 or negative number^0, it is better to define N_{ij} as
    $$
    N_{ij}=(-1)^{e_{ij0}} \prod_{\substack{1\le t\le d\\e_{ijt>0}}}z_t^{e_{ijt}},
    $$
    with parameters
    $\mathbf{z} = (z_{1}, \dots, z_{d}) \in \mathbb{R}^{d}$,
    we have $\PlGH(N) \leq^{T}_{p} \PlGH(M)$.
\end{lemma}
\begin{proof}
    For any $n \geq 1$ and graph $G = (V, E)$, recall from
    \cref{equation: thickeningEqn} that
    \begin{equation}\label{Z_M-expression}
        Z_{M}(\mathbf{T_{n}}G)
        = \sum_{x \in X(G)} x^{n} \cdot \text{\#}_{M}(G, x),
    \end{equation}
    where $X(G)$ and $\text{\#}_{M}(G, x)$ are as in
    \cref{equation: thickeningX(G),equation: thickeningNumMappings}.
    Since $X(G)$ is a finite set with $|X(G)| \leq
    |\mathcal{P}_{q^{2}}(|E|)| \leq |E|^{O(1)}$,
    querying $\PlGH(M)$ for $n \in [|X(G)|]$
    gives us a full-rank Vandermonde system in the variables
    $\{\text{\#}_{M}(G, x)\}_{x \in X(G)}$.
    Hence all $\text{\#}_{M}(G, x)$ can be computed in polynomial time.

    Next, each $x \in X(G)$ can be written as
    $x = \prod_{i, j \in [q]} M_{ij}^{k_{ij}}$
    for some $\mathbf{k} = (k_{ij}) \in \mathcal{P}_{q^{2}}(|E|)$.
    Since the entries of $M$ are generated by $\{g_t\}_{t \in [d]}$,
    every such $x$ admits a unique representation
    $$x = (-1)^{e^{x}_{0}} \cdot 
    g_{1}^{e^{x}_{1}} \cdots g_{d}^{e^{x}_{d}},
    \qquad e^{x}_{t} = \sum_{i, j \in [q]} k_{ij} \cdot e_{ijt},$$
    with $e^{x}_{t} \in \mathbb{Z}_{\ge 0}$ since
    all $e_{ijt}, k_{ij} \ge 0$.

    Define a map $\widehat{y}: X(G) \to \mathbb{R}$ by
    %$$\widehat{y}(x) = (-1)^{e^{x}_{0}} \cdot z_{1}^{e^{x}_{1}} \cdots z_{d}^{e^{x}_{d}},$$
    %%%%ZZT: avoid 0^0:
    $$\widehat{y}(x) = (-1)^{e^{x}_{0}} \prod_{\substack{1\le t\le d\\e_t^x>0}}z_t^{e_t^x},$$
    where
    $\mathbf{z} = (z_{1}, \dots, z_{d})$ are the parameters
    of $N \in \Sym{q}{}{\mathbb{R}}$.
    For any $\mathbf{k} \in \mathcal{P}_{q^{2}}(|E|)$,
    this definition ensures that
    $$\widehat{y} \left(\prod_{i, j \in [q]} M_{ij}^{k_{ij}}\right)
    = \prod_{i, j \in [q]} N_{ij}^{k_{ij}}.$$
    Let
    $$Y(G) = \left\{\prod_{i, j \in [q]} N_{ij}^{k_{ij}}
    \ \middle|\ \mathbf{k} \in \mathcal{P}_{q^{2}}(|E|) \right\}.$$

    For any labeling $\sigma : V \to [q]$, define
    $k_{ij} = |\{(u,v) \in E : \sigma(u)=i, \sigma(v)=j\}|$.
    Then
    $$\widehat{y}\left( \prod_{(u,v)\in E}
    M_{\sigma(u)\sigma(v)} \right)
    = \widehat{y}\left(\prod_{i, j \in [q]}M_{ij}^{k_{ij}}\right)
    = \prod_{i, j \in [q]} N_{ij}^{k_{ij}}
    = \prod_{(u,v)\in E} N_{\sigma(u)\sigma(v)}.$$
    Hence, for each $y \in Y(G)$,
    $$\left\{ \sigma : V \to [q] \ \middle|\ 
    \prod_{(u,v)\in E} N_{\sigma(u)\sigma(v)} = y \right\}
    = \bigsqcup_{\substack{x \in X(G):\\ \widehat{y}(x)=y}}
    \left\{ \sigma : V \to [q] \ \middle|\ 
    \prod_{(u,v)\in E} M_{\sigma(u)\sigma(v)} = x \right\}.$$
    Therefore, we obtain
    $$\text{\#}_{N}(G, y)
    = \sum_{\substack{x \in X(G):\\ \widehat{y}(x) = y}}
    \text{\#}_{M}(G, x).$$
    Using the values $\text{\#}_{M}(G, x)$ already determined,
    we can compute
    \begin{align*}
        Z_{N}(G)
        &= \sum_{y \in Y(G)} y \cdot \text{\#}_{N}(G, y)
        = \sum_{x \in X(G)} \widehat{y}(x) \cdot \text{\#}_{M}(G, x).
    \end{align*}
    Therefore, $Z_{N}(G)$ can be computed in polynomial time
    given oracle access to $\PlGH(M)$, implying
    $\PlGH(N) \leq^{T}_{p} \PlGH(M)$.
\end{proof}

We are now ready to prove
\cref{lemma: thickeningLemma}.

\thickeningLemma*
\begin{proof}
    Since $M \in \Sym{q}{}{\mathbb{R}_{\neq 0}}$
    is generated by $\{g_t\}_{t \in [d]}$, we know that
    there exist integers $e_{ijt} \in \mathbb{Z}$ such that
    for all $i, j \in [q]$,
    $$M_{ij} = (-1)^{e_{ij0}} \cdot
    g_{1}^{e_{ij1}} \cdots g_{d}^{e_{ijd}}.$$
    Now, we let $e_{t}^{*} = \min_{i, j \in [q]}e_{ijt}$
    for all $t \in [d]$, and let
    $c = \prod_{t \in [d]}g_{t}^{-e_{t}^{*}}$.
    From \cref{lemma: MequivalentCM}, we see that
    $\PlGH(M) \equiv_p^T \PlGH(cM)$.
    Moreover, we also see that the entries of the matrix
    $cM$ are also generated by $\{g_t\}_{t \in [d]}$
    such that for all $i, j \in [q]$,
    $$(cM)_{ij} = (-1)^{e_{ij0}} \cdot
    g_{1}^{e_{ij1} - e_{1}^{*}} \cdots g_{d}^{e_{ijd} - e_{d}^{*}}.$$
    From \cref{definition: mathcalTLarge,lemma: thickeningGeneralLemma},
    it follows that
    $\PlGH(\mathcal{T}(M; \mathbf{z})) \leq^{T}_{p} \PlGH(cM)
    \equiv_p^T \PlGH(M)$ for all
    $\mathbf{z} = (z_{1}, \dots, z_{d}) \in \mathbb{R}^{d}$.
\end{proof}

\cref{lemma: simpleThickening}
now follows as a corollary:
\simpleThickening*
\begin{proof}
    Let the entries of $M \in \Sym{q}{}{\mathbb{R}_{> 0}}$
    be generated by some $\{g_t\}_{t \in [d]}$,
    such that $M_{ij} = g_{1}^{e_{ij1}} \cdots g_{d}^{e_{ijd}}$
    for all $i, j \in [q]$.
    Let $e_{t}^{*} = \min_{i, j \in [q]}e_{ijt}$ for $t \in [d]$.
    Then, we note from \cref{definition: mathcalTLarge}, that
    $$\mathcal{T}(M; (g_{1}^{\theta}, \dots, g_{d}^{\theta}))_{ij}
    = (g_{1}^{e_{ij1}} \cdots g_{d}^{e_{ijt}})^{\theta} \cdot
    (g_{1}^{-e_{1}^{*}}\cdots g_{d}^{-e_{d}^{*}})^{\theta}
    = (M_{ij})^{\theta} \cdot c^{\theta},$$
    for $c = (g_{1}^{-e_{1}^{*}}\cdots g_{d}^{-e_{d}^{*}})$.
    Importantly, $c$ is independent of all $i, j \in [q]$.
    So, $\mathcal{T}(M; (g_{1}^{\theta}, \dots, g_{d}^{\theta}))
    = c^{\theta} \cdot \mathcal{T}_{M}(\theta)$ for all
    $\theta \in \mathbb{R}$.
    Therefore, \cref{definition: mathcalTSmall,lemma: thickeningLemma,lemma: MequivalentCM}
    imply that for all $\theta \in \mathbb{R}$,
    $$\PlGH(\mathcal{T}_{M}(\theta)) \equiv_p^T 
    \PlGH(\mathcal{T}(M; (g_{1}^{\theta}, \dots, g_{d}^{\theta})))
    \leq^{T}_{p} \PlGH(M).$$
\end{proof}

%\newpage
\subsection{Stretching Gadgets}\label{appendix: stretchingAppendix}

In this section, we give omitted proofs in \Autoref{sec: stretchingGadgets}.
In particular, we prove \cref{lemma: latticeInterpolation} (first proved in \cite{cai2023complexity}) from the perspective of ``conformal lattice interpolation" in \cite{CaiFST26}.
We first recall \cref{definition: latticeSet}.

\latticeSet*

\begin{lemma}[Conformal lattice interpolation \cite{CaiFST26}]\label{lm:conformal lattice interpolation}
    Let $k\in \Z_+$, $\mathbf{x}=(x_1,x_2,\ldots,x_k)\in (\C^*)^k$ and  $\mathbf{y}=(y_1,y_2,\ldots,y_k)\in (\C^*)^k$.
    Suppose 
    ${\mathfrak{L}}(x_1,x_2\ldots,x_k) \subseteq {\mathfrak{L}}(y_1,y_2\ldots,y_k)$.
    Given the numbers $$N_l(x_1, x_2, \ldots, x_k) := \sum\limits_{\substack{j_1,j_2,\dots, j_k \geq 0 \\ j_1+j_2+\dots+j_k= m}} (x_1^{j_1}x_2^{j_2}\cdots x_k^{j_k})^l z_{j_1,j_2,\dots,j_k}$$ for $l = 1,2,\dots,  \binom{m+k-1}{k-1}$,  we can compute $N_1(y_1, y_2, \ldots, y_k)$ in time polynomial in $m$.
\end{lemma}
\begin{remark}
    There is a slightly difference between the statement of \cref{lm:conformal lattice interpolation} and the conformal lattice interpolation in \cite{CaiFST26}, but the proofs are essentially the same.
\end{remark}

\latticeInterpolation*
\begin{proof}

Given a matrix $M \in \Sym{q}{pd}{\mathbb{R}}$,
we know that there exists an orthonormal matrix
(not necessarily unique) 
of unit eigenvectors $H$
with entries $H_{ij} \in \mathbb{R}$,
and a diagonal matrix of eigenvalues
$D = \diag(\lambda_{1}, \dots, \lambda_{q})$,
such that $M = HDH^{\tt{T}}$.
Therefore,
$$(\mathbf{S_{\ell}}(M))_{ij} = 
(M^{\ell})_{ij} = (H_{i1}H_{j1}) \lambda_{1}^{\ell} + \dots +
(H_{iq}H_{jq}) \lambda_{q}^{\ell},$$
for all $i, j \in [q]$, and all $\ell \geq 1$.

Given a planar graph $G=(V,E)$, we may construct $\mathbf{S}_{\ell}G$ for $1\le \ell \le \binom{|E|+q-1}{q-1}=|E|^{O(1)}$ in polynomial time, and query
\begin{equation*}
    \begin{aligned}
        Z_{M}(\mathbf{S}_{\ell}G)=Z_{\mathbf{S}_{\ell}(M)}(G)
    =&\sum_{\sigma:V\to [q]}\prod_{(u,v)\in E}\sum_{k=1}^qH_{\sigma(u)k}H_{\sigma(v)k}\lambda_k^\ell\\
    =&\sum_{\substack{j_1,j_2,\ldots,j_q\ge 0\\j_1+j_2+\cdots+j_q=|E|}}z_{j_1,j_2,\ldots,j_q}(\lambda_{j_1}\lambda_{j_2}\cdots\lambda_{j_q})^\ell.
    \end{aligned}
\end{equation*}
Here the coefficients $z_{j_1,j_2,\ldots,j_q}$ only depend on $G$ and $H$, but not depend on $\lambda_i$'s.
Similarly, $Z_{H\Delta H^{\tt T}}(G)$ has the expression
\begin{equation*}
    \begin{aligned}
        Z_{H\Delta H^{\tt T}}(G)=&\sum_{\sigma:V\to [q]}\prod_{(u,v)\in E}\sum_{k=1}^qH_{\sigma(u)k}H_{\sigma(v)k}\Delta_k\\
        =&\sum_{\substack{j_1,j_2,\ldots,j_q\ge 0\\j_1+j_2+\cdots+j_q=|E|}}z_{j_1,j_2,\ldots,j_q}\Delta_{j_1}\Delta_{j_2}\cdots\Delta_{j_q}
    \end{aligned}
\end{equation*}
By \cref{lm:conformal lattice interpolation}, oracle access to $Z_{M}(\mathbf{S}_{\ell}G)$ for $1\le \ell \le \binom{|E|+q-1}{q-1}$ recovers $Z_{H\Delta H^{\tt T}}(G)$.
Therefore, $\PlGH(H\Delta H^{\tt T})\le_p^T\PlGH(M).$
\end{proof}

Before we prove \cref{lemma: stretchingLemma}, we first recall \cref{definition: mathcalS} and prove that $\mathcal{S}_M$ is well-defined.
\DefMathcalS*

\begin{restatable}{lemma}{RealPowerOfMatrix}\label{lemma: real power of pd matrix}
For $M \in \Sym{q}{pd}{\mathbb{R}}$, if $M = U \Delta U^*$ for any unitary $U$ and diagonal $\Delta$, then
the diagonal entries are the (real positive) eigenvalues of $M$, and  $\mathcal{S}_{M}(\theta) = 
U \Delta^{\theta} U^*$. 
In particular,  $\mathcal{S}_M(\theta)$  is independent of the choice of the orthogonal decomposition $HDH^{\tt{T}}$ in \cref{definition: mathcalS}. 
%of $H$ or $D$.
   % The definition of $\mathcal{S}_M(\theta)$ is well-defined, as it is determined uniquely by $M$ and does not vary with the choice of $H$ or $D$.
\end{restatable}
\begin{proof}
    It is clear that the diagonal entries of $\Delta$ are the real positive eigenvalues of $M.$
    Suppose $\lambda_1>\lambda_2>\ldots>\lambda_k>0$, where  $1 \le k \le q$, are all distinct eigenvalues of $M$.
    Define the linear map $\mathcal{M}:\C^q \to \C^q, \mathbf{v}\mapsto M\mathbf{v}.$
    Then $\mathcal{M}$ is also diagonalizable, and $\lambda_1>\lambda_2>\ldots>\lambda_k>0$ are all distinct eigenvalues of $\mathcal{M}.$
    Let $E_i=\mathrm{Ker}(\mathcal{M}-\lambda_i\mathcal{I})\subseteq \C^q$ be the eigenspace corresponding to $\lambda_i$, for every $1\le i\le k$.
    Since $\mathcal{M}$ is diagonalizable, $\C^q=\bigoplus_{i=1}^kE_i$ is a direct sum.
    Thus, every $\mathbf{v}\in\C^q$ can be uniquely decomposed as $\mathbf{v}=\sum_{i=1}^k \mathbf{v}_i$, where $\mathbf{v}_i\in E_i.$
    Define the linear map $\mathcal{M}^\theta:\C^q\to \C^q: \mathbf{v}\mapsto \sum_{i=1}^k \lambda_i^\theta \mathbf{v}_i$.
    Suppose $U=(\mathbf{u}_1,\ldots,\mathbf{u}_{q})$ is a unitary matrix that diagonalizes $M=U\Delta U^*$.
    Hence $\{\mathbf{u}_1,\ldots,\mathbf{u}_{q}\}$ form a basis of $\C^q$.
    For every $1\le i\le q$, $\mathbf{u}_i$ belongs to some eigenspace $E_{d_i}$, where $1\le d_i\le k.$
    Then $\mathcal{M}^\theta (\mathbf{u}_i)=\lambda_{d_i}^\theta \mathbf{u}_i$.
    Thus, the matrix representation of the linear map $\mathcal{M}^\theta$ under the basis $\{\mathbf{u}_1,\ldots,\mathbf{u}_{q}\}$
    is $\mathrm{diag}(\lambda_{d_1}^\theta,\lambda_{d_2}^\theta,\ldots,\lambda_{d_q}^\theta)=\Delta^\theta$.
    Since $I=UU^*$, $U^*$ is the transition matrix from $\{\mathbf{u}_1,\ldots,\mathbf{u}_{q}\}$ to the standard orthonormal basis $\{\mathbf{e}_1,\ldots,\mathbf{e}_q\}.$ 
    Thus, the matrix representation of the linear map $\mathcal{M}^\theta$ under the standard basis $\{\mathbf{e}_1,\ldots,\mathbf{e}_{q}\}$ is $(U^*)^{-1}\Delta^\theta U^*=U\Delta^\theta U^*$.
    Since the matrix representation of $\mathcal{M}^\theta$ under the standard basis is unique, we see that $U \Delta^\theta U^*=HD^\theta H^{\tt{T}}=\mathcal{S}_M(\theta)$ is independent of choice of the (real) orthogonal decomposition $HDH^{\tt{T}}$.
\end{proof}

\cref{lemma: stretchingLemma} now follows as
an immediate corollary of \cref{lemma: latticeInterpolation}.

\stretchingLemma*
\begin{proof}
    Let $D = \diag(\lambda_{1}, \dots, \lambda_{q})$
    be the eigenvalues of $M$.
    Since $M \in \Sym{q}{pd}{\mathbb{R}}$, it follows that
    $\lambda_{i} > 0$ for all $i \in [q]$, and
    $D^{\theta} = \diag(\lambda_{1}^{\theta},
    \dots, \lambda_{q}^{\theta})$ is well-defined
    for all $\theta \in \mathbb{R}$,
    %_{\neq 0}$,
    with $\lambda_{i}^{\theta} > 0$ for all $i \in [q]$.
    Moreover, for any $\mathbf{k}=(k_1,\ldots,k_{q})\in \Z^q$,
    %for $\theta \neq 0$ and $\lambda_i > 0$ for all $i \in [q]$,
    \begin{align*}
        \mathbf{k} \in \mathfrak{L}(\lambda_{1}, \dots, \lambda_{q})
        &\iff \lambda_{1}^{k_{1}} \cdots \lambda_{q}^{k_{q}} = 1 \\
        &\implies(\lambda_{1}^{k_{1}} \cdots
        \lambda_{q}^{k_{q}})^{\theta}= 1\\
        &\iff (\lambda_{1}^{\theta})^{k_{1}} \cdots
        (\lambda_{q}^{\theta})^{k_{q}} = 1
        \iff \mathbf{k} \in \mathfrak{L}(\lambda_{1}^{\theta},
        \dots, \lambda_{q}^{\theta}).
    \end{align*}
    So, $\mathfrak{L}(\lambda_{1}, \dots, \lambda_{q}) 
    \subseteq \mathfrak{L}(\lambda_{1}^{\theta}, \dots, \lambda_{q}^{\theta})$.
    It then immediately follows from
    \cref{definition: mathcalS,lemma: latticeInterpolation} that
    $\PlGH(\mathcal{S}_{M}(\theta)) \leq^{T}_{p} \PlGH(M)$,
    for all $\theta \in \mathbb{R}$.
\end{proof}

\section{Hardness Criteria}\label{Appendix: I+uuT}

In this section, we give all omitted proofs from \cref{sec: I+uuT}.

\GeneratingSetofRankOnePerturbation*
\begin{proof}
    By the proof of \cref{lemma: generatingSet}, we may choose a generating set $\mathcal{G}_1=\{g_t\mid t\in [d]\}$ of the off-diagonals $\{u_iu_j\mid 1\le i<j\le q\}$ such that $g_1,g_2\ldots,g_d$ are multiplicative independent, i.e.,  $\widehat{\mathfrak{L}}(g_1,\ldots, g_d)=\{\mathbf{0}\}.$
    Let $\mathcal{G}_2=\{v_i\mid 1\le i\le r\}$ be distinct elements among $\{u_i\mid 1\le i\le q\}.$
    Let $c>1$ be transcendental over the field $F:=\mathbb{Q}(g_1\ldots, g_t,v_1,\ldots,v_r)$.
   
    We next show that elements in $\mathcal{G}_1\cup \mathcal{G}_2$ are multiplicative independent.
    Suppose for a contradiction that $$\prod_{i=1}^r(c+v_i^2)^{x_i}\prod_{t=1}^d g_t^{y_t}=1$$ for some  $(x_1,\ldots, x_r,y_1,\ldots, y_d)\in \mathbb{Z}^{d+r}\setminus \{\mathbf{0}\}.$
    If $(x_1,\ldots, x_r)=\mathbf{0}$, then $(y_1,\ldots, y_d)\neq \mathbf{0}$ and $\prod_{t=1}^r g_t^{y_t}=1.$
    So $(y_1,\ldots, y_d)\in \widehat{\mathfrak{L}}(g_1,\ldots, g_d)$, contradicting that $\widehat{\mathfrak{L}}(g_1,\ldots, g_d)=\{\mathbf{0}\}.$
    Thus $(x_1,\ldots, x_r)\neq \mathbf{0}$.
    We have
    \begin{equation}\label{eq: generating set of I+uuT}
        \prod_{t=1}^d g_t^{y_t}\prod_{\substack{1\le i\le r\\x_i\ge 0}}(c+v_i^2)^{x_i}
        =\prod_{\substack{1\le i\le r\\x_i< 0}}(c+v_i^2)^{-x_i}
    \end{equation}
    Let $$P(X)=\prod_{t=1}^d g_t^{y_t}\prod_{\substack{1\le i\le r\\x_i\ge 0}}(X+v_i^2)^{x_i}
    -\prod_{\substack{1\le i\le r\\x_i< 0}}(X+v_i^2)^{-x_i}.$$
    We have $P(X)\in F[X].$
    Also, $P(X)$ is a non-zero polynomial since $v_i$'s are discint and $(x_1,\ldots,x_r)\neq \mathbf{0}$.
    But $c$ is transcendental over $F$, so $P(c)\neq 0$.
    This contradicts \eqref{eq: generating set of I+uuT}.
    Therefore, elements in $\mathcal{G}_1\cup \mathcal{G}_2$ are multiplicative independent.
    
    Consequently, for every $i\in[q]$, $c+u_i^2$ has a unique multiplicative expression by elements in $\mathcal{G}_2$.
    So $\mathcal{G}_2$ is a generating set of $\{c+u_i^2\mid 1\le i\le q\}.$
    It is also easy to see that every entry of $cI+\mathbf{uu}^{\tt T}$ has a unique multiplicative expression by elements in $\mathcal{G}_1\cup \mathcal{G}_2$.
    So $\mathcal{G}_1\cup \mathcal{G}_2$ is a generating set of all entries in $cI+\mathbf{uu}^{\tt T}$.
\end{proof}
\section{Circulant Matrix of Prime Order}\label{appendix: Circulant}
In this section, we give all omitted proofs from \cref{sec: circulant}.

\PropertiesOfCirculant*
\begin{proof}
    Let $\psi_m=\sum_{j=0}^{q-1}c_j\zeta_q^{mj}$ for $0\le m\le q-1.$ 
    By \cref{lm: evs of circulant matrix}, $(\psi_0,\psi_1,\ldots,\psi_{q-1})$ are eigenvalues of $C$.
    Since $C$ is positive definite, $\psi_m>0$ for every $0\le m\le q-1$.
    Notice that $(\psi_0,\psi_1,\ldots,\psi_{q-1})$ is a Fourier transformation of $(c_0,c_1,\ldots, c_{q-1})$, we have $c_j=\frac{1}{q}\sum_{m=0}^{q-1}\psi_m \zeta_q^{-jm} \,(\forall 0\le j\le q-1)$ by the Fourier inversion formula.
    By the fact that $c_j,\psi_m\in\R$ for every $j,m\in[q]$, we have 
    $$
    c_j=\mathrm{Re}\left(\frac{1}{q}\sum_{m=0}^{q-1}\psi_m e^{\frac{2\pi \mathfrak{i} jm}{q}}\right)=\frac{1}{q}\sum_{m=0}^{q-1}\psi_m \cos\left(\frac{2\pi jm}{q}\right), ~~(\forall 0\le j\le q-1).
    $$
    Then 
    $$
    c_0-c_j=\frac{1}{q}\sum_{m=0}^{q-1}\psi_m\left(1-\cos\left(\frac{2\pi jm}{q}\right)\right)>0 , ~~(1\le j\le q-1).
    $$
    Similarly, 
    $$
    \psi_m=\mathrm{Re}\left(\sum_{j=0}^{q-1}c_je^{\frac{-2\pi \mathfrak{i} mj}{q}}\right)=\sum_{j=0}^{q-1}c_j \cos\left(\frac{2\pi mj}{q}\right), ~~(\forall m\in[q])
    $$
    and $\psi_0>\psi_m$ for every $1\le m\le q-1$.
\end{proof}

\begin{remark}\label{rmk: lattice of Potts}
    Recall that $\potts{q}{x}=\mathcal{C}(x,1\ldots,1)$.
    We note that by \cref{lemma: Potts Model},
    $\PlGH(\potts{q}{x})$ is \#P-hard for every $q\ge 3$ and real $x\neq 1$.
    By \cref{lm: evs of circulant matrix}, the eigenvalues of $\potts{q}{x}$ are 
    $
    \psi_0=x+q-1,$
    and
    $\psi_{m}=x+\sum_{j=1}^{q-1}\zeta_q^{mj}=x-1,\,\forall 1\le m\le q-1.
    $
    It follows from \cref{def: lattice of potts,definition: latticeSet} that for every $x>1$,
    $$
    \LPotts{q}\subseteq \mathfrak{L}(\psi_0,\psi_1,\ldots,\psi_{q-1}).
    $$
    For $x>1$, $\frac{x+q-1}{x-1}>1$ is not a root of unity.
    Thus, for any $\mathbf{x}=(x_0,x_1,\ldots,x_{q-1})\in \Z$ such that $\sum_{m=0}^{q-1}x_m=0$, 
    $$
    \prod_{m=0}^{q-1} \psi_m^{x_m}=1 \implies\left(\frac{x+q-1}{x-1}\right)^{x_0}=1 \implies x_0=0.
    $$
    That is,
    $$
     \mathfrak{L}(\psi_0,\psi_1,\ldots,\psi_{q-1})\subseteq \LPotts{q}.
    $$
    Therefore, we conclude that 
    $
     \mathfrak{L}(\psi_0,\psi_1,\ldots,\psi_{q-1})= \LPotts{q}
    $
    for any $x>1.$
\end{remark}

\section{Tensor Product of 2 by 2 Matrices}
\label{Appendix: Tensor of 2by2}
%The following two lemmas follow from direct calculation (see \autoref{Appendix: Tensor of 2by2}):
In this section, we give all omitted or abbreviated proofs from \cref{sec: tensor of 2 by 2}.

\RhoGammaInZeroOne*
\begin{proof}
    For simplicity, let $\lambda=\lambda(A)$, and define $\mu,\rho,\gamma$ similarly.
    Since $A\in\Sym{2}{pd}{\Rp}$, we have $0<\mu<\lambda$ by \cref{def: 2 by 2 matrix parameters}.
    Thus $\rho=\frac{\mu}{\lambda}\in (0,1).$
    Notice that $\lambda$ and $\mu$ are roots of $\varphi_A(x)=x^2-(a+c)x+ac-b^2$.
    We have $\varphi_A(a)=-b^2<0$, so $\mu<a<\lambda$, which implies $\gamma>0$.
    Also $\lambda+\mu=a+c<2a$, which implies $\gamma<1$.    
\end{proof}
%\vspace{-.2in}

\begin{restatable}{lemma}{BThetaAndCTheta}\label{lm: b(theta) and c(theta) after stretching}
    Let $A=\left[\begin{smallmatrix}
    1 & b\\
    b & c
\end{smallmatrix}\right]\in \Sym{2}{pd}{\Rp}$ with $c<1$, and $\rho=\rho(A),\gamma=\gamma(A)$.
Then up to a scalar normalization, $\mathcal{S}_A(\theta)=\left[\begin{smallmatrix}
    1 & b(\theta)\\
    b(\theta) & c(\theta)
\end{smallmatrix}\right]$ for any $\theta\in \R,$ where 
$
b(\theta)=\frac{\sqrt{\gamma}(1-\rho^\theta)}{1+\gamma\rho^\theta}
$
and
$
c(\theta)=\frac{\rho^\theta+\gamma}{1+\gamma\rho^\theta}.
$
\end{restatable}

%\BThetaAndCTheta*
\begin{proof}
By \cref{lemma: real power of pd matrix}, $A^\theta=\mathcal{S}_A(\theta)$ is a well defined (matrix-valued) analytic function of $\theta$.
Let $A(\theta):=\frac{A^\theta}{A^\theta_{11}}=\left[\begin{smallmatrix}
    1 & b(\theta)\\
    b(\theta) & c(\theta)
\end{smallmatrix}\right]$ be the normalized matrix of $A^\theta$.
In the following, $\lambda=\lambda(A)$ and $\mu=\mu(A)$.
Note that $\lambda,\mu$ are roots of $\varphi_A(x)=x^2-(1+c)x+c-b^2$.
Using Sylvester's formula for $2\times 2$ matrix, we have
$$A^\theta = \lambda^\theta \frac{A - \mu I}{\lambda - \mu} + \mu^\theta \frac{A - \lambda I}{\mu - \lambda},\theta>0.$$
Then $$A^{\theta}_{11}=\frac{\lambda^{\theta}(1-\mu)-\mu^\theta(1-\lambda)}{\lambda-\mu},~~
A^{\theta}_{22}=\frac{\lambda^{\theta}(c-\mu)-\mu^\theta(c-\lambda)}{\lambda-\mu},$$
and 
\begin{equation}\label{eq: c after stretching}
    \begin{aligned}
        c(\theta)=&\frac{A_{22}^\theta}{A_{11}^\theta}
        =\frac{\lambda^{\theta}(c-\mu)-\mu^\theta(c-\lambda)}{\lambda^{\theta}(1-\mu)-\mu^\theta(1-\lambda)}\\
=&\frac{\lambda^\theta\times\frac{b^2}{1-\mu}-\mu^\theta\times\frac{b^2}{1-\lambda}}{\lambda^{\theta}(1-\mu)-\mu^\theta(1-\lambda)}\\
    =&\frac{b^2}{(1-\lambda)(1-\mu)}\times\frac{\lambda^\theta(1-\lambda)-\mu^\theta(1-\mu)}{\lambda^{\theta}(1-\mu)-\mu^\theta(1-\lambda)}\\
    =&\frac{\mu^\theta(1-\mu)-\lambda^\theta(1-\lambda)}{\lambda^{\theta}(1-\mu)-\mu^\theta(1-\lambda)}\\
    =&\frac{\rho^\theta+\gamma}{1+\gamma\rho^\theta},
    \end{aligned}
\end{equation}
Similarly, we can calculate
\begin{equation}\label{eq: b after stretching}
    \begin{aligned}        b(\theta)=&\frac{A^\theta_{12}}{A^\theta_{11}}=\frac{\lambda^\theta b-\mu^\theta b}{\lambda^{\theta}(1-\mu)-\mu^\theta(1-\lambda)}\\
        =&\frac{(\lambda^\theta-\mu^\theta)\sqrt{-(1-\lambda)(1-\mu)}}{\lambda^{\theta}(1-\mu)-\mu^\theta(1-\lambda)}\\
        =&\frac{\sqrt{\gamma}(1-\rho^\theta)}{1+\gamma\rho^\theta}.
    \end{aligned}
\end{equation}
\end{proof}

\TensorProductFirstReduction*
\begin{proof}
    Suppose $M\in \Sym{q}{pd}{\Rp}$ and  $M=\bigotimes_{i=1}^s A_i$, where $A_i\in\R^{2\times 2}$.
    We first prove that $A_i$ is symmetric for every $1\le i\le s.$
    Since $M$ is symmetric, $\bigotimes_{i=1}^s(A_i)^{\tt{T}}= M^{\tt{T}}=M=\bigotimes_{i=1}^s A_i$.
    This implies $A_i=t_i A_i^{\tt{T}}$ for some $t_i\in \R$, for every $1\le i\le s.$
    So $A_i=t_i(t_iA_i^{\tt{T}})^{\tt{T}}=t_i^2A_i$.
    Since $M$ has all entries positive, $A_i\neq O$.
    This implies $t_i=1$ or $t_i=-1$, for every $1\le i\le s.$
    If for some $1\le i\le s$, $t_i=-1$, then $A_i$ is skew symmetric and the diagonal entries of $A_i$ are 0.
    This contradicts that all entris of $M$ are positive.
    So $t_i=1$ for every $1\le i\le s$.
    i.e., $A_i$ is symmetric for every $1\le i\le s$.
    
    So far we may assume $A_i=\left[\begin{smallmatrix}
        a_i & b_i\\
        b_i & c_i
    \end{smallmatrix}\right]\in\Sym{2}{}{\R_{\neq 0}}.$
    Since $M$ has all entries positive, the entries within any $A_i$ must have the same sign. 
    That is, for every $1\le i\le s$, either $A_i\in \Sym{q}{}{\Rp}$ or $A_i\in \Sym{q}{}{\R_{<0}}$.
    Note that for an arbitrary graph $G$, 
    \begin{equation}\label{eq: Z_M(G) is product of ZA_i(G)}
        Z_M(G)=\prod_{i=1}^s Z_{A_i}(G).
    \end{equation}
    Also note that $Z_{-A_i}(G)=(-1)^{|E(G)|}Z_{A_i}(G)$.
    So if $A_i\in\Sym{q}{}{\R_{<0}}$, we may replace $A_i$ by $-A_i$ without changing the complexity of $\PlGH(M)$.
    Thus, we may assume that $A_i\in \Sym{2}{}{\Rp}$ for every $1\le i\le s$.
    Since the eigenvalues of $M$ are products of those of $A_i$, and they are all positive, the two eigenvalues of every $A_i$ are either both positive or both negative.
    Also, the sum of eigenvalues of $A_i$ is equal to $\mathrm{Tr}(A_i)>0$, so they are both positive.
    Thus, we may assume that $A_i\in \Sym{2}{pd}{\Rp}$ for every $1\le i\le s$.

    It is known that if $A\in \Sym{2}{pd}{\Rp}$ has equal diagonal entries, then $A$ is matchgate transformable \cite{cai2017complexitybook}.
    So if $a_i=c_i$ for every $1\le i\le s$, then $M$ is a tensor product of matchgates transformable signatures, and $\PlGH(M)$ is polynomial time solvable by Valiant's holographic algorithm plus the FKT algorithm. %(\textcolor{red}{Maybe need a citation here.})
    Thus, to prove \cref{thm: hardness for 2 by 2 tensor positive definite}, it suffices to prove that $\PlGH(M)$ is \#P-hard in the case that there exists some $A_i$ which has distinct diagonal entries.
    If there exists some $1\le i\le s$ such that $a_i=c_i$, then $Z_{A_i}(G)>0$ is easy to compute. 
    Thus, to compute $Z_M(G)$ is equivalent to compute the product of those terms on the RHS of \cref{eq: Z_M(G) is product of ZA_i(G)} corresponding to non-matchgate $A_i$ .
    So, we may assume that for every $1\le i\le s$, $a_i\neq c_i$.
    Also, let $A^r_i=\left[\begin{smallmatrix}
        c_i & b_i\\
        b_i & a_i
    \end{smallmatrix}\right]$, we have $Z_{A_i}(G)=Z_{A^r_i}(G)$ by flipping 0 and 1  in the assignment of vertices of the input graph $G$.
    So, we may assume that $a_i>c_i$ for every $1\le i\le s$.
    Finally, by \cref{lemma: MequivalentCM}, we can normalize $a_i=1$ for every $1\le i\le s$.
    
    So far, we may assume $M=\bigotimes_{i=1}^sA_i$, where $A_i=\left[\begin{smallmatrix}
        1 & b_i\\
        b_i & c_i
    \end{smallmatrix}\right]\in \Sym{2}{pd}{\Rp}$ and $c_i<1$ for every $1\le i\le s$.
    Next we do stretching to realize $b_i>c_i$.
    By \cref{lm: b(theta) and c(theta) after stretching}, for every $1\le i\le s$, $\lim_{\theta\to + \infty}b_i(\theta)=\sqrt{\gamma_i}$ and $\lim_{\theta\to + \infty}c_i(\theta)=\gamma_i$.
    Since $\gamma_i\in(0,1)$ by \cref{lm: rho and gamma range}, there exists some $\theta>0$, such that $1>b_i(\theta)>c_i(\theta)>0$ $(\forall1\le i\le s)$.
    Combining \cref{lm: b(theta) and c(theta) after stretching}, \cref{lemma: MequivalentCM} and \cref{lemma: stretchingLemma}, $\PlGH(\bigotimes_{i=1}^s A_i(\theta))\equiv_p^T\PlGH(\bigotimes_{i=1}^s \mathcal{S}_{A_i}(\theta))\equiv_p^T \PlGH(\mathcal{S}_M(\theta))\le_p^T \PlGH(M)$.
    Notice that if each $A_i \in \Sym{2}{pd}{\Rp}$, $\mathcal{S}_{A_i}(\theta)$ is also positive definite.
    That is, $c_i(\theta)>b_i(\theta)^2$.
    %%% JYC I edited:
    % stretching has no affect on the positive definiteness,
    So we may replace $M$ by $\mathcal{S}_M(\theta)$, and $A_i$ by $A_i(\theta)$, and assuming $1>b_i>c_i>b_i^2>0$, for every $1\le i\le s$. 
\end{proof}

\RhoDistinctByThickening*
\begin{proof}
    Recall that $B_i=A_{p_i}$ and $\mathcal{T}_{B_i}(n)=\left[\begin{smallmatrix}
    1 & b_{p_i}^n\\
    b_{p_i}^n & c_{p_i}^n
    \end{smallmatrix}\right]$, and $(b_{p_i},c_{p_i})$ are all distinct for $1\le i\le m.$
    We only need to prove that for arbitrary $1\le i<j
    \le s$, there exists $N_{ij}\in \mathbb{Z}_+$ such that for every $n>N_{ij}$, $\rho(\mathcal{T}_{B_i}(n))\neq \rho(\mathcal{T}_{B_j}(n))$.
    The lemma follows by taking an arbitrary $n>\max_{1\le i<j\le s} N_{ij}$.
    We prove this by looking at the Taylor expansion for each $\rho(\mathcal{T}_{B_i}(n))$ as $n\to +\infty$. 
    For simplicity, we temporarily omit subscripts in the following.
    Using $\sqrt{1+x}=1+\frac{1}{2}x -\frac{1}{8}x^2+\Theta(x^3)$ as $x\to 0$ and $c>b^2$, we have 
    \begin{equation*}
        \begin{aligned}
            \sqrt{\Delta_n}:=
            &\sqrt{(1-c^n)^2+4b^{2n}}=\sqrt{1-2c^n+c^{2n}+4b^{2n}}\\=&
            1+\frac{1}{2}(-2c^n+c^{2n}+4b^{2n})-\frac{1}{8}(-2c^n+c^{2n}+4b^{2n})^2+
            \Theta(c^{3n})\\
            =&
            1-c^n+\frac{1}{2}c^{2n}+2b^{2n}-\frac{1}{8}(4c^{2n}+16b^{4n}-16c^n b^{2n})+
            \Theta(c^{3n})\\
            =&
            1-c^n+2b^{2n}+2c^n b^{2n}-2b^{4n}+
            \Theta(c^{3n}),\text{ as $n\to\infty$.}\\
        \end{aligned}
    \end{equation*}  
    By assumption $b>c$, 
    \begin{equation}
        \sqrt{\Delta_n}=1-c^n+2b^{2n}+\Theta(c^nb^{2n}),\text{ as $n\to\infty$.}
    \end{equation}
    Then
    $$
    \rho(\mathcal{T}_{B}(n))=\frac{1+c^n-\sqrt{\Delta_n}}{1+c^n+\sqrt{\Delta_n}}=\frac{c^n-b^{2n}+\Theta(c^{n}b^{2n})}{1+b^{2n}+\Theta(c^{n}b^{2n})}, \text{ as $n\to\infty$.}
    $$
    Using $\frac{1}{1+x}=1-x+\Theta(x^2)$ as $x\to 0$ and $c>b^2$,
    \begin{equation}
        \rho(\mathcal{T}_{B}(n))=(c^n-b^{2n}+\Theta(c^{n}b^{2n}))\times(1-b^{2n}+\Theta(c^{n}b^{2n}))=c^n-b^{2n}+\Theta(c^nb^{2n}), \text{ as $n\to\infty$.}
    \end{equation}

    Now fix arbitrary $1\le i<j\le m$ and consider $\rho(\mathcal{T}_{B_i}(n))$ and $\rho(\mathcal{T}_{B_j}(n))$.
    If $c_{p_i}\neq c_{p_j}$, say $c_{p_i}>c_{p_j}$, then $\rho(\mathcal{T}_{B_j}(n))=o(\rho(\mathcal{T}_{B_i}(n)))$.
    Hence there exists $N_{ij}\in \Zp$, such that for every $n>N_{ij}$, $\rho(\mathcal{T}_{B_i}(n))\neq \rho(\mathcal{T}_{B_j}(n)).$
    If $c_{p_i}=c_{p_j}$, then $b_{p_i}\neq b_{p_j}$ by the assumption that all $B_i$  are distinct.
    Say $b_{p_i}>b_{p_j}$, then $\rho(\mathcal{T}_{B_j}(n))-\rho(\mathcal{T}_{B_i}(n))=b_{p_i}^{2n}-b_{p_j}^{2n}+\Theta(c_{p_i}^nb_{p_i}^{2n}+c_{p_j}^nb_j^{2n})=b_{p_i}^{2n}+o(b_{p_i}^{2n})$.
    Hence there exists $N_{ij}\in \Zp$, such that for every $n>N_{ij}$,  $\rho(\mathcal{T}_{B_i}(n))\neq \rho(\mathcal{T}_{B_j}(n))$.
\end{proof}

In the following, we let ${\mathfrak{L}}(\{\lambda_T\}_{T\subseteq\{0,1\}^s})$ denote the lattice ${\mathfrak{L}}(\lambda_{00\ldots0}, \ldots, \lambda_{11\ldots1})$, where the $s$ indices are ordered lexicographically.

\begin{lemma}\label{lm: lattice translation}
    Let $s\in \Zp$, $(\rho_1,\rho_2,\ldots,\rho_s),(\rho'_1,\rho'_2,\ldots,\rho'_s)\in \R_{\neq 0}^s$, and
    \begin{equation}\label{eq: lambda and mu}
        \lambda_{T}=\prod_{i\in T} \rho_i, \quad \lambda'_{T}=\prod_{i\in T} \rho'_i, \quad \text{for }T\subseteq \{0,1\}^s.
    \end{equation}
    Then
     $$\widehat{\mathfrak{L}}(\rho_1,\ldots,\rho_t)\subseteq \widehat{\mathfrak{L}}(\rho'_1,\ldots,\rho'_t) \Longrightarrow {\mathfrak{L}}(\{\lambda_T\}_{T\subseteq\{0,1\}^s})\subseteq {\mathfrak{L}}(\{\lambda'_T\}_{T\subseteq\{0,1\}^s}).$$ 
\end{lemma}
\begin{proof}
    Assume $\widehat{\mathfrak{L}}(\rho_1,\ldots,\rho_t)\subseteq \widehat{\mathfrak{L}}(\rho'_1,\ldots,\rho'_t).$
    For any $\mathbf{x}=(x_T)_{T\subseteq \{0,1\}^s}\in {\mathfrak{L}}(\{\lambda_T\}_{T\subseteq\{0,1\}^s})$, we have 
    $$
    \sum_{T\subseteq \{0,1\}^s}x_{T}=0,
    \text{ and }
    \prod_{T\subseteq\{0,1\}^s}\lambda_{T}^{x_{T}}=1.
    $$
    Substituting $\lambda_T=\prod_{i \in T}\rho_i$ into the second equation above yields
    \begin{equation*}
        \prod_{i=1}^s\rho_i^{y_i}=1,\text{ where } y_i=\sum_{i\in T\subseteq\{0,1\}^s}x_T.
    \end{equation*}
    So $\mathbf{y}=(y_1,\ldots,y_s)\in \widehat{\mathfrak{L}}(\rho_1,\ldots,\rho_s)\subseteq \widehat{\mathfrak{L}}(\rho_1',\ldots,\rho_s')$.
    It follows that $\prod_{i=1}^s(\rho_i')^{y_i}=1$, which is equivalent to
    $
    \prod_{T\subseteq \{0,1\}^s}(\lambda'_T)^{x_T}=1.
    $
    Therefore, $\mathbf{x}\in {\mathfrak{L}}(\{\lambda'_T\}_{T\subseteq\{0,1\}^s}).$
\end{proof}

\ReducedLatticeLemma*
\begin{proof}
Suppose $B_i$ appears $t_i$ times in $A_1,A_2,\ldots,A_s$, $1\le i\le m.$
Let $N=\bigotimes_{i=1}^m B_i^{\otimes t_i}.$
Then for any planar graph $G$, $Z_N(G)=\prod_{i=1}^m (Z_{B_i}(G))^{t_i}=Z_M(G)$.
So $\PlGH(N)\equiv_p^T \PlGH(M).$
By assumption $\widehat{\mathfrak{L}}(\rho(B_1),\rho(B_2),\ldots,\rho(B_s))=\{\mathbf{0}\},$ we have
     \begin{equation}\label{eq: lattice of repeated mu}
        \begin{aligned}
            &\widehat{\mathfrak{L}}(\rho(B_1),\ldots,\rho(B_1),\rho(B_2),\ldots,\rho(B_2),\ldots,\rho(B_m),\ldots,\rho(B_m)) = ~\bigoplus_{i=1}^m\widehat{\mathfrak{L}}(\rho(B_i) \ldots, \rho(B_i))\\
            \subseteq ~ &
            \widehat{\mathfrak{L}}(\rho(B_1) \ldots, \rho(B_1))\oplus\bigoplus_{i=2}^s \widehat{\mathfrak{L}}(1,\ldots,1) = ~\widehat{\mathfrak{L}}(\rho(B_1),\ldots,\rho(B_1),1,\ldots,1,\ldots,1,\ldots,1),\\
        \end{aligned}
    \end{equation}
    where $\rho(B_i)$ repeats $t_i$ times for every $1\le i\le m$.   
    We diagonalize $B_i$ by $B_i=U_i\Lambda_iU_i^{\tt{T}}$, where $U_i\in O(2)$ and 
    $$
    \Lambda_i=\mathrm{diag}(\lambda(B_i),\mu(B_i)).
    $$
    Then $N=U \Lambda U^{\tt{T}}$, where $U=\bigotimes_{i=1}^m U_i^{\otimes t_i}$ and $\Lambda=\bigotimes_{i=1}^m \Lambda_i^{\otimes t_i}$.
    We do a re-normalization by letting
    $C_i=U_i\mathrm{diag}(1,\rho(B_i))U_i^{\tt{T}}$ and $N_1=\bigotimes_{i=1}^m C_i^{\otimes p_i}$.
    Then $N_1=U\Delta U^{\tt{T}}$, where $$\Delta=\bigotimes_{i=1}^m \mathrm{diag}(1,\rho(B_i))^{\otimes t_i}.$$
    By \cref{lemma: MequivalentCM}, $\PlGH(N_1)\equiv_p^T \PlGH(N).$
    Let $N'=B_1^{\otimes t_1}\otimes I^{\otimes t_2}\otimes\cdots\otimes I^{\otimes t_m}$.
    Then $N'= U \Lambda'U^{\tt{T}}$, where $\Lambda'=\Lambda_1^{\otimes t_1}\otimes I^{\otimes t_2}\otimes\cdots \otimes I^{\otimes t_m}.$
    We also do a re-normalization to $N'$ by letting $N_1'=U\Delta' U^{\tt{T}}$,
    where 
    $$
    \Delta'=\mathrm{diag}(1,\rho(B_1))^{\otimes t_1}\otimes I^{\otimes t_2}\otimes\cdots\otimes I^{\otimes t_m}.
    $$
    Also by \cref{lemma: MequivalentCM}, $\PlGH(N_1')\equiv_p^T \PlGH(N')$.
    
    The eigenvalues of $N_1$ are given by
    \begin{equation}\label{eq: lambda mu repeated A_i}
        \lambda_{(T_1,T_2,\ldots,T_m)}=\prod_{i=1}^m \rho(B_i)^{|T_i|}, \text{ for all } T_i\subseteq [t_i], 1\le i\le m.
    \end{equation}
    Similarly, the eigenvalues of $N_1'$ are
    \begin{equation}\label{eq: lambda' mu' repeated A_i}
        \lambda_{(T_1,T_2,\ldots,T_m)}=\rho(B_1)^{|T_1|}\prod_{i=2}^m 1^{|T_i|}, \text{ for all } T_i\subseteq [t_i], 1\le i\le m.
    \end{equation}
    Combining \cref{eq: lattice of repeated mu} \cref{eq: lambda mu repeated A_i}, \cref{eq: lambda' mu' repeated A_i} and \cref{lm: lattice translation}, 
        $$
        \mathfrak{L}(\{\lambda_{(T_1,T_2,\ldots,T_m)}\}_{T_i\subseteq [t_i],1\le i\le m})
        \subseteq 
        \mathfrak{L}(\{\lambda'_{(T_1,T_2,\ldots,T_m)}\}_{T_i\subseteq [t_i],1\le i\le m})
        $$
        Then $\PlGH(N_1')\le_p^T \PlGH(N_1)$ by \cref{lemma: latticeInterpolation}.
        Hence $\PlGH(N')\le_p^T \PlGH(N)\equiv_p^T\PlGH(M)$.
        Since $\PlGH(N')\equiv_p^T \PlGH(B_1)$ and $B_1$ is not matchgate, $\PlGH(N')$ is \#P-hard, and thus $\PlGH(M)$ is \#P-hard.
\end{proof}

We will now prove 
\cref{thm: hardness for 2 by 2 tensor nonnegative},
but we shall first need to prove a few lemmas.

\begin{lemma}\label{lm: square makes positive definite}
    Let $A\in \Sym{2}{F}{\R_{>0}}$.
    Then $A^{2}\in \Sym{2}{pd}{\Rp}$.
    Moreover, $A_{11} = A_{22}$ if and only if $(A^{2})_{11}=(A^{2})_{22}$.
\end{lemma}
\begin{proof}
    Since $A$ is symmetric and full rank,
    its eigenvalues are nonzero reals.
    Hence the eigenvalues of $A^{2}$ are strictly positive,
    so $A^{2}$ is positive definite.
    Writing
    $A = \left[\begin{smallmatrix}
    a & b\\
    b & c\end{smallmatrix}\right]$
    with $a, b, c > 0$,
    $$
        (A^{2})_{11} - (A^{2})_{22} = 
        (a^{2} + b^{2}) - (c^{2} + b^{2}) = a^{2} - c^{2},
    $$
    proving the second claim.
\end{proof}

\begin{lemma}\label{lm: strictly positive full rank case}
    Let $M = \bigotimes_{i=1}^{s} A_{i}$ with
    $A_{i} \in \Sym{2}{F}{\R_{>0}}$.
    Then $\PlGH(M)$ is \#P-hard unless
    $(A_{i})_{11}=(A_i)_{22}$ for every $1\le i\le s$,
    in which case, $\PlGH(M)$ is polynomial-time tractable.
\end{lemma}
\begin{proof}
    By \cref{lm: square makes positive definite},
    $M^{2} = \bigotimes_{i=1}^{s} A_{i}^{2} \in \Sym{q}{pd}{\Rp}$.
    Also $A_{i}^{2}$ has equal diagonal entries
    if and only if $A_{i}$ does.
    By stretching,
    $\PlGH(M^{2})\le_{p}^{T} \PlGH(M)$.
    The lemma follows from
    \cref{thm: hardness for 2 by 2 tensor positive definite}.
\end{proof}

\begin{lemma}\label{lm: rank zero one factors}
    Let $M = \bigotimes_{i=1}^{s} A_{i}$ with 
    $A_{i} \in \Sym{2}{}{\R_{>0}}$.
    Let $I = \{i:\det(A_i) \neq 0\}$.
    Then
    $$
        \PlGH(M) \equiv_{p}^{T}
        \PlGH\left(\bigotimes_{i \in I} A_{i}\right).
    $$
\end{lemma}

\begin{proof}
    If $\det(A_{i}) = 0$, then
    $A_{i} = \mathbf{u}\mathbf{u}^{\tt{T}}$
    for some $\mathbf{u} = (x, y)^{\tt{T}}\in (\R_{>0})^{2}$.
    For any graph $G = (V,E)$,
    $$
        Z_{A_{i}}(G) = \prod_{v \in V}
        \big(x^{\deg(v)} + y^{\deg(v)}\big)>0,
    $$
    which is computable in polynomial time.
    Using
    $
        Z_M(G) = \prod_{i=1}^{s} Z_{A_{i}}(G),
    $
    the claim follows.
\end{proof}

\begin{lemma}\label{lm: strictly positive case}
    Let $M = \bigotimes_{i=1}^{s} A_{i}$ with
    $A_{i}\in \Sym{2}{}{\R_{>0}}$.
    Let $I = \{i \in[s]: \det(A_{i}) \neq 0\}$.
    Then, $\PlGH(M)$ is \#P-hard unless $(A_{i})_{11}=(A_i)_{22}$ for every $i \in I$,
    in which case, $\PlGH(M)$ is polynomial-time tractable.
\end{lemma}
\begin{proof}
    If $I = \emptyset$, then $\PlGH(M)$ is trivially
    polynomial-time tractable, and the proof is complete.
    So, let's assume $I \neq \emptyset$.
    Let $M' := \bigotimes_{i \in I}A_{i}$.
    By \cref{lm: rank zero one factors},
    $\PlGH(M) \equiv_{p}^{T} \PlGH(M')$.
    The claim follows from
    \cref{lm: strictly positive full rank case}.
\end{proof}

\begin{lemma}\label{lm: diagonal factor direct sum}
    Let $M = \bigotimes_{i=1}^{s} A_{i}$ with
    $A_{i} \in \Sym{2}{}{\R_{\geq 0}}$.
    If $(A_{j})_{12} = 0$ for some $j$,
    then $M$ is a direct sum of two blocks
    $$M \cong (\alpha M')\oplus (\beta M'),$$
    where $A_{j} = \mathrm{diag}(\alpha,\beta)$
    and $M' = \bigotimes_{i\neq j}A_{i}$.
    Moreover, $\PlGH(M')\equiv_{p}^{T} \PlGH(M)$.
\end{lemma}
\begin{proof}
    Since $A_{j}$ is diagonal,
    for any choice of assignments on the other
    $s - 1$ tensor factors, 
    the contribution of the $j$-th factor is either
    $\alpha$ or $\beta$.
    Hence $M$ decomposes as a direct sum
    of two copies of $M'$, scaled by
    $\alpha$ and $\beta$, respectively,
    i.e., $M = (\alpha M') \oplus (\beta M')$.

    We now apply a standard result on direct sums
    \cite[Lemma~4.6]{cai2013graph}
    (whose proof uses the first pinning lemma,
    \cite[Lemma~4.1]{cai2013graph}):
    if $M = A \oplus B$, then $\PlGH(M)$ is \#P-hard
    if at least one of $\PlGH(A)$ or $\PlGH(B)$
    is \#P-hard, and $\PlGH(M)$ is polynomial-time tractable,
    if both $\PlGH(A)$ and $\PlGH(B)$ are polynomial-time tractable.
    Applying this to $M = (\alpha M') \oplus (\beta M')$,
    we conclude that $\PlGH(M') \equiv_{p}^{T} \PlGH(M)$.
\end{proof}

\HardnessForTensorTwoByTwoNonnegative*
\begin{proof}
    Let $M = \bigotimes_{i=1}^{s} A_{i}$ with $A_{i} \in \Sym{2}{}{\R_{\ge 0}}$.
    By repeatedly applying \cref{lm: diagonal factor direct sum},
    we may assume $(A_{i})_{12} > 0$ for all $i$.
    Then $I = \{i \in[s]: \det(A_i)\neq 0\}$.

    If $I = \emptyset$,
    then $\PlGH(M)$ is polynomial-time tractable by
    \cref{lm: rank zero one factors}.
    Hence assume $I \neq \emptyset$.

    Let $M' = \bigotimes_{i\in I}A_{i}$.
    By \cref{lm: rank zero one factors},
    $\PlGH(M)\equiv_p^T \PlGH(M')$.

    Squaring yields
    $(M')^{2} = \bigotimes_{i\in I} A_{i}^{2}$.
    By \cref{lm: square makes positive definite},
    each $A_{i}^{2}\in \Sym{2}{pd}{\Rp}$ and
    $(A_{i}^2)_{11} = (A_{i}^2)_{22}$
    if and only if $(A_{i})_{11} = (A_{i})_{22}$.
    By stretching,
    $\PlGH((M')^{2}) \leq_{p}^{T} \PlGH(M').$

    If $(A_{i})_{11}=(A_i)_{22}$ for every $i \in I$,
    then $M'$ (and hence $M$) is polynomial-time tractable.
    Otherwise, by
    \cref{thm: hardness for 2 by 2 tensor positive definite},
    $\PlGH((M')^{2})$ is \#P-hard,
    hence $\PlGH(M') \leq_{p}^{T} \PlGH(M)$ is \#P-hard.
\end{proof}

\bibliographystyle{plain}
\bibliography{ref}

@article{Circulant_Matrices,
author = {Gray, Robert},
year = {2001},
month = {10},
pages = {},
title = {Toeplitz and Circulant Matrices: A Review},
volume = {2},
journal = {Foundations and Trends® in Communications and Information Theory},
doi = {10.1561/0100000006}
}

@book{jacobson1985basic,
  title={Basic Algebra},
  author={Jacobson, N.},
  number={v. 2},
  isbn={9780716719335},
  lccn={84025836},
  series={Basic Algebra},
  url={https://books.google.com/books?id=oNmDSAAACAAJ},
  year={1985},
  publisher={W.H. Freeman}
}

@inproceedings{dyer2000complexity,
  title={The complexity of counting graph homomorphisms (extended abstract)},
  author={Dyer, Martin and Greenhill, Catherine},
  editor={David B. Shmoys},
  booktitle={Proceedings of the Eleventh Annual {ACM-SIAM} Symposium on Discrete Algorithms, January 9-11, 2000, San Francisco, CA, {USA}},
  pages={246--255},
  publisher={{ACM/SIAM}},
  year={2000}
}

@article{bulatov2005complexity,
  title={The complexity of partition functions},
  author={Bulatov, Andrei and Grohe, Martin},
  journal={Theoretical Computer Science},
  volume={348},
  number={2-3},
  pages={148--186},
  year={2005},
  publisher={Elsevier}
}

@article{goldberg2010complexity,
  title={A complexity dichotomy for partition functions with mixed signs},
  author={Goldberg, Leslie Ann and Grohe, Martin and Jerrum, Mark and Thurley, Marc},
  journal={SIAM Journal on Computing},
  volume={39},
  number={7},
  pages={3336--3402},
  year={2010},
  publisher={SIAM}
}

@article{cai2013graph,
  title={Graph homomorphisms with complex values: A dichotomy theorem},
  author={Cai, Jin-Yi and Chen, Xi and Lu, Pinyan},
  journal={SIAM Journal on Computing},
  volume={42},
  number={3},
  pages={924--1029},
  year={2013},
  publisher={SIAM}
}

@article{vertigan2005computational,
  title={The computational complexity of Tutte invariants for planar graphs},
  author={Vertigan, Dirk},
  journal={SIAM Journal on Computing},
  volume={35},
  number={3},
  pages={690--712},
  year={2005},
  publisher={SIAM}
}

@book{lovasz2012large,
  title={Large networks and graph limits},
  author={Lov{\'a}sz, L{\'a}szl{\'o}},
  volume={60},
  year={2012},
  publisher={American Mathematical Soc.}
}

@article{bulatov2013complexity,
  title={The complexity of the counting constraint satisfaction problem},
  author={Bulatov, Andrei A.},
  journal={Journal of the ACM (JACM)},
  volume={60},
  number={5},
  pages={1--41},
  year={2013},
  publisher={ACM New York, NY, USA}
}

@inproceedings{dyer2010complexity,
  title={On the complexity of \# CSP},
  author={Dyer, Martin E. and Richerby, David M.},
  booktitle={Proceedings of the forty-second ACM symposium on Theory of computing},
  pages={725--734},
  year={2010}
}

@InProceedings{dyer2011decidability,
  author =	{Dyer, Martin E. and Richerby, David M.},
  title =	{{The \# CSP Dichotomy is Decidable}},
  booktitle =	{28th International Symposium on Theoretical Aspects of Computer Science (STACS 2011) },
  pages =	{261--272},
  year =	{2011},
  volume =	{9},
  editor =	{Thomas Schwentick and Christoph D{\"u}rr},
  publisher =	{Schloss Dagstuhl--Leibniz-Zentrum fuer Informatik},
  address =	{Dagstuhl, Germany},
  doi =		{10.4230/LIPIcs.STACS.2011.261}
}

@article{dyer2013complexity,
  volume={42},
  number={3},
  author={Dyer, Martin E. and Richerby, David M.},
  title = {An effective dichotomy for the counting constraint satisfaction problem},
  publisher = {Society for Industrial and Applied Mathematics},
  year = {2013},
  journal = {SIAM Journal on Computing},
  pages = {1245 -- 1274},
}

@article{bulatov2012complexity,
  title={The complexity of weighted and unweighted \# CSP},
  author={Bulatov, Andrei and Dyer, Martin and Goldberg, Leslie Ann and Jalsenius, Markus and Jerrum, Mark and Richerby, David},
  journal={Journal of Computer and System Sciences},
  volume={78},
  number={2},
  pages={681--688},
  year={2012},
  publisher={Elsevier}
}

@article{cai2017complexity,
author = {Cai, Jin-Yi and Chen, Xi},
year = {2017},
month = {06},
pages = {1-39},
title = {Complexity of Counting {CSP} with Complex Weights},
volume = {64},
journal = {Journal of the ACM},
doi = {10.1145/2822891}
}

@article{cai2016nonnegative,
  title={Nonnegative weighted \#{CSP}: an effective complexity dichotomy},
  author={Cai, Jin-Yi and Chen, Xi and Lu, Pinyan},
  journal={SIAM Journal on Computing},
  volume={45},
  number={6},
  pages={2177--2198},
  year={2016},
  publisher={SIAM}
}

@article{valiant2008holographic,
  title={Holographic algorithms},
  author={Valiant, Leslie G},
  journal={SIAM Journal on Computing},
  volume={37},
  number={5},
  pages={1565--1594},
  year={2008},
  publisher={SIAM}
}

@article{valiant1979complexity,
  title={The complexity of computing the permanent},
  author={Valiant, Leslie G.},
  journal={Theoretical computer science},
  volume={8},
  number={2},
  pages={189--201},
  year={1979},
  publisher={Elsevier}
}

@article{kasteleyn1961statistics,
  title={The statistics of dimers on a lattice: I. The number of dimer arrangements on a quadratic lattice},
  author={Kasteleyn, Pieter W.},
  journal={Physica},
  volume={27},
  number={12},
  pages={1209--1225},
  year={1961},
  publisher={Elsevier}
}

@article{temperley1961dimer,
  title={Dimer problem in statistical mechanics-an exact result},
  author={Temperley, Harold NV. and Fisher, Michael E.},
  journal={Philosophical Magazine},
  volume={6},
  number={68},
  pages={1061--1063},
  year={1961},
  publisher={Taylor \& Francis}
}

@article{kasteleyn1963dimer,
  title={Dimer statistics and phase transitions},
  author={Kasteleyn, Pieter W.},
  journal={Journal of Mathematical Physics},
  volume={4},
  number={2},
  pages={287--293},
  year={1963},
  publisher={American Institute of Physics}
}

@article{kasteleyn1967graph,
  title={Graph theory and crystal physics},
  author={Kasteleyn, Pieter},
  journal={Graph theory and theoretical physics},
  pages={43--110},
  year={1967},
  publisher={Academic Press}
}

@inproceedings{cai2009holant,
  title={Holant problems and counting CSP},
  author={Cai, Jin-Yi and Lu, Pinyan and Xia, Mingji},
  booktitle={Proceedings of the forty-first annual ACM symposium on Theory of computing},
  pages={715--724},
  year={2009}
}

@article{cai2016complete,
  title={A complete dichotomy rises from the capture of vanishing signatures},
  author={Cai, Jin-Yi and Guo, Heng and Williams, Tyson},
  journal={SIAM Journal on Computing},
  volume={45},
  number={5},
  pages={1671--1728},
  year={2016},
  publisher={SIAM}
}

@article{cai2019holographic,
  title={Holographic algorithm with Matchgates is universal for planar \# CSP over boolean domain},
  author={Cai, Jin-Yi and Fu, Zhiguo},
  journal={SIAM Journal on Computing},
  volume={51},
  number={2},
  pages={STOC17--50},
  year={2019},
  publisher={SIAM}
}

@article{banica2005quantum,
  title={Quantum automorphism groups of homogeneous graphs},
  author={Banica, Teodor},
  journal={Journal of Functional Analysis},
  volume={224},
  number={2},
  pages={243--280},
  year={2005},
  publisher={Elsevier}
}

@article{atserias2019quantum,
  title={Quantum and non-signalling graph isomorphisms},
  author={Atserias, Albert and Man{\v{c}}inska, Laura and Roberson, David E. and {\v{S}}{\'a}mal, Robert and Severini, Simone and Varvitsiotis, Antonios},
  journal={Journal of Combinatorial Theory, Series B},
  volume={136},
  pages={289--328},
  year={2019},
  publisher={Elsevier}
}

@article{chassaniol2019study,
  title={Study of quantum symmetries for vertex-transitive graphs using intertwiner spaces},
  author={Chassaniol, Arthur},
  journal={arXiv preprint arXiv:1904.00455},
  year={2019}
}

@article{wang1995free,
  title={Free products of compact quantum groups},
  author={Wang, Shuzhou},
  journal={Communications in Mathematical Physics},
  volume={167},
  number={3},
  pages={671--692},
  year={1995},
  publisher={Springer}
}

@article{wang1998quantum,
  title={Quantum symmetry groups of finite spaces},
  author={Wang, Shuzhou},
  journal={arXiv preprint math/9807091},
  year={1998}
}

@article{kar2026npa,
  title={NPA hierarchy for quantum isomorphism and homomorphism indistinguishability},
  author={Kar, Prem Nigam and Roberson, David E and Seppelt, Tim and Zeman, Peter},
  journal={Quantum},
  volume={10},
  pages={1989},
  year={2026},
  publisher={Verein zur F{\"o}rderung des Open Access Publizierens in den Quantenwissenschaften}
}

@article{lupini2020nonlocal,
  title={Nonlocal games and quantum permutation groups},
  author={Lupini, Martino and Man{\v{c}}inska, Laura and Roberson, David E.},
  journal={Journal of Functional Analysis},
  volume={279},
  number={5},
  pages={108592},
  year={2020},
  publisher={Elsevier}
}

@inproceedings{manvcinska2020quantum,
  title={Quantum isomorphism is equivalent to equality of homomorphism counts from planar graphs},
  author={Man{\v{c}}inska, Laura and Roberson, David E.},
  booktitle={2020 IEEE 61st Annual Symposium on Foundations of Computer Science (FOCS)},
  pages={661--672},
  year={2020},
  organization={IEEE}
}

@article{lovasz1967operations,
  title={Operations with structures},
  author={Lov{\'a}sz, L{\'a}szl{\'o}},
  journal={Acta Math. Acad. Sci. Hungar},
  volume={18},
  number={3-4},
  pages={321--328},
  year={1967}
}

@inproceedings{backens2017new,
  author={Miriam Backens},
  title={A New Holant Dichotomy Inspired by Quantum Computation},
  booktitle={44th International Colloquium on Automata, Languages, and Programming, {ICALP} 2017, July 10-14, 2017, Warsaw, Poland},
  series={LIPIcs},
  volume={80},
  pages={16:1--16:14},
  publisher={Schloss Dagstuhl - Leibniz-Zentrum f{\"{u}}r Informatik},
  year= {2017}
}

@inproceedings{backens2018complete,
  author={Miriam Backens},
  title={A Complete Dichotomy for Complex-Valued Holant\^{}c},
  booktitle={45th International Colloquium on Automata, Languages, and Programming, {ICALP} 2018, July 9-13, 2018, Prague, Czech Republic},
  series={LIPIcs},
  volume={107},
  pages={12:1--12:14},
  publisher={Schloss Dagstuhl - Leibniz-Zentrum f{\"{u}}r Informatik},
  year={2018}
}

@article{yang2022local,
  title={Local holographic transformations: tractability and hardness},
  author={Yang, Peng and Fu, Zhiguo},
  journal={Frontiers of Computer Science},
  volume={17},
  number={2},
  pages={1--11},
  year={2022},
  publisher={Springer}
}

@article{fu2019blockwise,
  title={On blockwise symmetric matchgate signatures and higher domain \# CSP},
  author={Fu, Zhiguo and Yang, Fengqin and Yin, Minghao},
  journal={Information and Computation},
  volume={264},
  pages={1--11},
  year={2019},
  publisher={Elsevier}
}

@article{fu2014holographic,
  title={Holographic algorithms on bases of rank 2},
  author={Fu, Zhiguo and Yang, Fengqin},
  journal={Information Processing Letters},
  volume={114},
  number={11},
  pages={585--590},
  year={2014},
  publisher={Elsevier}
}

@inproceedings{cai2023complexity,
  title={The complexity of counting planar graph homomorphisms of domain size 3},
  author={Cai, Jin-Yi and Maran, Ashwin},
  booktitle={Proceedings of the 55th Annual ACM Symposium on Theory of Computing},
  pages={1285--1297},
  year={2023}
}

@book{hell2004graphs,
  title={Graphs and homomorphisms},
  author={Hell, Pavol and Nesetril, Jaroslav},
  volume={28},
  year={2004},
  publisher={OUP Oxford}
}

@book{cai2017complexitybook,
  title={Complexity dichotomies for counting problems: Volume 1, Boolean domain},
  author={Cai, Jin-Yi and Chen, Xi},
  year={2017},
  publisher={Cambridge University Press}
}

@article{slofstra2019set, 
    title={The Set of Quantum Correlations is not Closed}, 
    volume={7}, 
    journal={Forum of Mathematics, Pi}, 
    author={Slofstra, William}, 
    year={2019}
}

@article{cai2024polynomial,
  title={Polynomial and analytic methods for classifying complexity of planar graph homomorphisms},
  author={Cai, Jin-Yi and Maran, Ashwin},
  journal={arXiv preprint arXiv:2412.17122},
  year={2024}
}

@article{cai2026dichotomy,
      title={Planar Graph Homomorphisms: A Dichotomy and a Barrier from Quantum Groups}, 
      author={Jin-Yi Cai and Ashwin Maran and Ben Young},
      year={2026},
      journal={arXiv preprint arXiv:2601.23198},
      url={https://arxiv.org/abs/2601.23198}, 
}

@article{bichon2003quantum,
  title={Quantum automorphism groups of finite graphs},
  author={Bichon, Julien},
  journal={Proceedings of the American Mathematical Society},
  volume={131},
  number={3},
  pages={665--673},
  year={2003}
}

@book{arveson1998invitation,
  title={An invitation to C*-algebras},
  author={Arveson, William},
  volume={39},
  year={1998},
  publisher={Springer Science \& Business Media}
}

@article{woronowicz1987compact,
  title={Compact matrix pseudogroups},
  author={Woronowicz, Stanis{\l}aw L},
  journal={Communications in Mathematical Physics},
  volume={111},
  number={4},
  pages={613--665},
  year={1987},
  publisher={Springer}
}

@article{FITZGERALD1977633,
title = {On fractional Hadamard powers of positive definite matrices},
journal = {Journal of Mathematical Analysis and Applications},
volume = {61},
number = {3},
pages = {633-642},
year = {1977},
issn = {0022-247X},
doi = {https://doi.org/10.1016/0022-247X(77)90167-6},
author = {Carl H FitzGerald and Roger A Horn}
}

@book{kato1995perturbation,
  title={Perturbation Theory for Linear Operators},
  author={Kato, Tosio},
  series={Classics in Mathematics},
  volume={132},
  year={1995},
  edition={2nd},
  publisher={Springer-Verlag},
  address={Berlin Heidelberg}
}

@inproceedings{CaiFST26,
  author       = {Jin{-}Yi Cai and
                  Austen Z. Fan and
                  Shuai Shao and
                  Zhuxiao Tang},
  editor       = {Aditya Bhaskara and
                  Artur Czumaj},
  title        = {New Planar Algorithms and a Full Complexity Classification of the
                  Eight-Vertex Model},
  booktitle    = {Proceedings of the 58th Annual {ACM} Symposium on Theory of Computing,
                  {STOC} 2026, Salt Lake City, UT, USA, June 22-26, 2026},
  pages        = {967--978},
  publisher    = {{ACM}},
  year         = {2026},
  url          = {https://doi.org/10.1145/3798129.3800810},
  doi          = {10.1145/3798129.3800810},
  bibsource    = {dblp computer science bibliography, https://dblp.org}
}

\end{document}